\documentclass[a4paper,11pt]{article}
\usepackage{amsthm,amsmath,amscd,amssymb,amsfonts,pstricks,pst-node,graphics}
\usepackage[small,nohug,head=littlevee]{diagrams}
\diagramstyle[labelstyle=\scriptstyle]
\usepackage{caption}
\providecommand*{\pderiv}[3][]{
        \frac{\partial^{#1}{#2}}
                {\partial {#3}^{#1}}}

\def\uu#1{u_{#1}}

\def\exnl{\Theta}
\def\cS{{\cal S}}
\def\cO{{\cal O}}

\newcommand{\liesl}{{\mathfrak{sl}}}

\def\N{{\mathbb{N}}}
\def\C{{\mathbb{C}}}
\def\Z{{\mathbb{Z}}}
\def\R{{\mathbb{R}}}

\def\im{\operatorname{Im}}

\def\exnl{\Theta}

\def\ad{{\rm ad}}
\def\s{{\bf S}}
\newcommand\A{{\mathcal F}}
\newcommand{\lieh}{{\mathfrak{h}}}

\newtheorem{Rem}{Remark}
\newtheorem{Not}{Notation}
\newtheorem{Def}{Definition}
\newtheorem{The}{Theorem}
\newtheorem{Pro}{Proposition}
\newtheorem{Lem}{Lemma}

\begin{document}
\title{Representations of $\liesl(2,\C)$ in category $\cO$ and master symmetries}
\author{Jing Ping Wang\\
School of Mathematics, Statistics \& Actuarial Science\\
University of Kent,
Canterbury, UK}
\date{}
\maketitle
\abstract{
In this paper, we first give a short account on the indecomposable
\(\liesl(2,\C)\) modules in the
Bernstein-Gelfand-Gelfand (BGG) category $\cO$. We show these modules naturally
arise for homogeneous integrable nonlinear evolutionary systems.
We then develop an approach to construct master symmetries for such
integrable systems. This
method naturally enables us to compute the hierarchy of time-dependent
symmetries. We finally illustrate
the method using both classical and new examples. We compare our
approach to the known existing methods used to construct master symmetries. For
the new integrable equations such as a Benjamin-Ono type equation, a new
integrable Davey-Stewartson type equation and two different
versions of (2+1)-dimensional generalised Volterra Chains,
we generate their conserved densities
using their master symmetries.
}

\vspace{0.2cm}

\vspace{0.2cm}

\section{Introduction}\label{sec1}

One of important hidden properties of integrable nonlinear evolution equations is the
existence of infinitely many commuting symmetries, which has been used as a
criterion to tackle the classification problems of integrable systems.  It has
produced fruitful results
(see for example review papers \cite{asy,mr86i:58070,mr89e:58062,mr93b:58070,mr95j:35010, mnw08, sw2009} and the references therein). 
One way to generate these commuting symmetries is to use the well-known recursion operators
\cite{AKNS74,mr58:25341}, which map a symmetry to a new symmetry. The Nijenhuis
property of recursion operators enables us to construct an abelian Lie
algebra of symmetries \cite{Fuc79, Mag80}. This paper is devoted to the
construction of master symmetries, which is an
alternative method to produce this hierarchy of mutually commuting symmetries.

A master symmetry is an evolutionary vector field $\tau$ whose adjoint action $\ad_{\tau}$ maps a symmetry to a new symmetry.
The concept of master symmetries was first introduced in \cite{mr83f:58039}, where the authors constructed
a master symmetry for the Benjamin-Ono equation and showed the equation has no
recursion operator of
polynomial type. This concept was further developed in
\cite{mr85m:58187, mr94j:58081}. A recent review on it can be found in \cite{asy}. For example, the
Burgers equation
$$ u_t=u_{xx}+u u_x$$ possesses a master symmetry
$$\tau=x(u_{xx}+u u_x) +\frac{1}{2} u^2 .$$
Starting from $a_0=u_x$,  the recurrence procedure $a_{n+1}=\ad_{\tau} a_n$ generates the hierarchy of commuting symmetries for 
the Burgers equation. This procedure is also called $\tau$-scheme \cite{mr94j:58081}.

There are several methods to construct master symmetries.  For homogeneous equations, master symmetries arise from applying their recursion operators to 
their corresponding scaling symmetry (see Definition \ref{def3}). They can also
be viewed as non-isospectral flows obtained from the Lax representation
\cite{OS86, ma92}.
Without having these extra structures at hand, master symmetries can be
constructed based on the existence of time dependent higher order symmetries
\cite{mr86c:58158, mr90j:58061}. The Burgers equation possesses infinite many $t$-dependent symmetries denoted by $S_n, n=1,2, \cdots$ and $S_n$ is polynomial
with respect to $t$ of degree $n$. We list a few of them:
\begin{eqnarray*}
&&S_1= t u_{x}+1;\\
&&S_2= t^2 (\uu{xx}+u \uu{x})+ t(x \uu{x}+u)+ x;\\
&&S_3= t^3 \left(\uu{xxx}+\frac{3}{2}u \uu{xx}+\frac{3}{2} \uu{x}^2+\frac{3}{4} u^2 \uu{x}
\right)
 + \frac{3}{2}  t^2\left(x \uu{xx}+ x u \uu{x}+2 \uu{x}+\frac{1}{2}u^2\right)\\
&&\qquad + \frac{3}{4}  t (x^2 \uu{x} +2 x u+2)+\frac{3}{4} x^2;\\
&&S_4= t^4 \left(\uu{xxxx}+2 u u_{xxx}+5 u_x u_{xx}+\frac{3}{2}u^2 \uu{xx}+3 u \uu{x}^2+\frac{1}{2} u^3 \uu{x}
\right)\\
&&\qquad +t^3 \left(x(2 \uu{xxx}+3 u \uu{xx}+3  \uu{x}^2+\frac{3}{2} u^2 \uu{x})+5 u_{xx}+6 u u_x+\frac{1}{2}u^3 \right)\\
&&\qquad+ \frac{3}{2}  t^2\left(x^2 \uu{xx}+ x^2 u \uu{x}+x u^2+4 x u_x+2
u\right)
 + \frac{1}{2}  t (x^3 \uu{x} +3 x^2 u+6 x) +\frac{1}{2} x^3.
\end{eqnarray*}
The master symmetry $\tau$ can be obtained from the coefficient of $t^2$ in $S_3$.

In order to systematically compute $t$-dependent symmetries of nonlinear
evolution equation, Sanders and the author observed the existence of
$\liesl(2, \C)$
in \cite{mr1874291} and applied the idea to some well-known equations such as
the Burgers and Kadomtsev--Petviashvili 
equations. For the Burgers equation, we let $e=u_x$, $h=2(x u_x+u)$ and $f=-(x^2 u_x+2x u +2)$, which are the $t$-coefficients in the above list. Then we have
\begin{eqnarray}\label{sl2b}
 [e,\ f]=h, \quad [h,\ e]=2 e, \quad [h,\ f]=-2 f,
\end{eqnarray}
where the Lie bracket defined by (\ref{liebra}), that is, these three elements
form an $\liesl(2, \C)$. Notice that 
$$[u_t, f]=[u_{xx}+u u_x, f]=4 \tau +8 u_x$$
is a master symmetry since $u_x$ commutes with all $a_n$. Later in
\cite{ffokas02}, Finkel and Fokas adapted it (without explicitly mentioning
$\liesl(2, \C)$) for 
equations with nonlocal master symmetries
such as the Sawada-Kotera equation. In this paper, we give justification of the method explaining where element $f$ is from and give rigorous conditions when 
$[f, u_t]$ is a master symmetries. We then apply the method to various examples
including some new $(2+1)$-dimensional equations presented in \cite{HNov13,
FNR13}.

Notice that elements $e$, $h$ and $f$ are all related to linear terms of $t$ in $S_i$, $i=1, 2, 3$. 
Let $w=\frac{1}{3} x^3 \uu{x} +x^2 u+2 x$ the linear term in $S_4$. Then we have $[e, w]=f$. We show 
that the vector space spanned by all linear terms of $S_n$, $n=1, 2, \cdots$ is
an infinite dimensional module of $\liesl(2, \C)$ in the 
Bernstein-Gelfand-Gelfand (BGG) category $\cO$. Thus, we are able to present
symmetries, master symmetries and their generalisations as elements
of Lie modules. We shall
refer this structure as the $\cO$-scheme for the corresponding integrable
equation. All vectors in this structure can be used to generate its (time
dependent) symmetries.

The arrangement of the paper is as follows: In section \ref{sec2} we give a short account on representations of \(\liesl(2,\C)\) in 
the BGG category $\cO$.
We then give definitions of symmetries and master symmetries for an evolutionary equation
and recall how to generate a hierarchy of conserved densities using master
symmetries in section \ref{sec3}.
The main theoretical results are presented in section \ref{sec4}. We first
show that for a homogeneous evolutionary equation, no matter it is integrable
or not, there is a natural $\liesl(2, \C)$ and further an infinite dimensional
module in the BGG category $\cO$\footnote{In recent paper \cite{troost2012}, the
author showed all indecomposable
$\liesl(2, \C)$ modules in the Bernstein-Gelfand-Gelfand category $\cO$ arises
as quantised phase space of physical models.}. We then give an approach to
construct master symmetries for such integrable equations and further obtain
its $\cO$-scheme. Finally, we show
that the elements in this infinite dimensional module enables us to generate
their time-dependent symmetries. In section \ref{sec5}, we use examples to
demonstrate how the proposed approach apply to various examples of both
$(1+1)$-dimensional and $(2+1)$-dimensional evolutionary equations. In the end,
we give a short discussion on further research in this
direction.

\section{The BBG category of \(\liesl(2,\C)\) modules}\label{sec2}
In this section we give a short account on the Bernstein-Gelfand-Gelfand (BGG) category $\cO$ of \(\liesl(2,\C)\) modules related to this paper.
We refer the reader to the book \cite{humph2008, mazor2010} for details.

The $\liesl(2,\C)$ algebra, or simply written as $\liesl(2)$,  is generated by three generators $e, f, h$ satisfying the relations
\begin{eqnarray}\label{sl2}
 [e,\ f]=h, \quad [h,\ e]=2 e, \quad [h,\ f]=-2 f.
\end{eqnarray}
Let $M$ be a $\liesl(2)$--module. For any  $\lambda\in \C$, we define a subspace $M_\lambda$ as
$$
M_\lambda=\left\{ v\in M \big{|} h \cdot v= \lambda v\right\}.
$$
If $\lambda$ is not an eigenvalue of the representation $h \cdot$, then $M_\lambda=\{ 0\}$. Whenever
$M_\lambda\neq \{ 0\}$, we call $\lambda$ a weight of $h$ in $M$ and $M_\lambda$ its associated weight space.

The objects of the BGG category $\cO$ of $\liesl(2)$--modules are $M$ such that
\begin{itemize}
 \item $M$ is finitely generated;
 \item $M$ is a weight module: $M$ is the direct sum of its weight spaces;
 \item All weight spaces of $M$ are finite dimensional.
\end{itemize}
It is clear that all finite dimensional modules lie in this category $\cO$. In fact, all modules in this category are the direct sum of the indecomposable
modules. The indecomposable modules inside the category $\cO$ of $\liesl(2)$ are classified \cite{humph2008}. Here we list out some basic results 
for these nonisomorphic indecomposable modules.
\subsection{The finite dimensional simple module $L(\lambda)$}
The finite dimensional simple module $L(\lambda)$ with $\lambda\in \N$ has dimension $\lambda+1$. It has $1$-dimensional weight spaces, with weights 
$\lambda, \lambda-2, . . . , -(\lambda-2)$ and $-\lambda$. One can choose basis vectors $v_i ( 0\leq i\leq \lambda)$ so that
\begin{eqnarray*}
&& f \cdot v_i=v_{i+1};\\
&& h \cdot v_i= (\lambda-2 i ) v_i;\\
&& e \cdot v_i= i (\lambda- i +1) v_{i-1};\\
&&e \cdot v_0=f \cdot v_{\lambda}=0.
\end{eqnarray*}
Here $v_0$ is a highest weight vector with highest weight $\lambda$. It is (up to scaling) the unique vector vanishing under the action of $e$.
\subsection{The Verma module $M(\lambda)$}
For any highest weight $\lambda\in \C$, the Verma module $M(\lambda)$ is defined as the module generated by the universal enveloping algebra acting on 
a highest weight vector. The weights of the Verma module $M(\lambda)$ are  $\lambda, \lambda-2, \lambda-4, \cdots$, each with multiplicity one.
One can choose basis vectors $v_i ( i\geq 0)$ so that
\begin{eqnarray*}
&& f \cdot v_i=v_{i+1};\\
&& h \cdot v_i= (\lambda-2 i ) v_i;\\
&& e \cdot v_i= i (\lambda- i +1) v_{i-1}.
\end{eqnarray*}
When $\lambda\in \N$, we have $e \cdot v_{\lambda+1}=0$. In this case the
maximal submodule of $M(\lambda)$ is $M(-\lambda-2)$, and  
the finite dimension simple module $L(\lambda)$ is the quotient space $M(\lambda)/M(-\lambda-2)$.

When $\lambda$ is not a positive integer, the Verma module $M(\lambda)$ is simple.
\subsection{The dual Verma module $M^{\vee}(\lambda)$}
Duality action within the category $\cO$ maps the Verma module $M(\lambda)$ to its dual $M^{\vee}(\lambda)$. 
This module has the property that one can reach the highest weight vector from any vector by action of $e$, which is dual to the property
of a Verma module $M(\lambda)$ that one can reach any vector from the highest weight vector by action of $f$.
For any highest weight $\lambda\in \C$, the weights of the dual Verma module are
the same as the Verma module.
One can choose basis vectors $v_i ( i\geq 0)$ so that
\begin{eqnarray*}
&& e \cdot v_i=v_{i-1};\\
&& h \cdot v_i= (\lambda-2 i ) v_i;\\
&& f \cdot v_i= (i+1) (\lambda- i) v_{i+1};\\
&& e \cdot v_0=0.
\end{eqnarray*}
When $\lambda\in \N$, we have $f \cdot v_{\lambda}=0$. In this case the maximal submodule of $M^{\vee}(\lambda)$ is $L(\lambda)$.
The quotient space $M^{\vee}(\lambda)/L(\lambda)$ is isomorphic to the Verma module $M(-\lambda-2)$).
\subsection{The projective module $P(-\lambda-2)$}
For $\lambda\in \N$, there are also non-trivial projective modules
$P(-\lambda-2)$.  It has weights $\lambda, \lambda-2$, $\cdots, -\lambda$ with
multiplicity one and 
the weights $-\lambda-2, -\lambda-4, \cdots$ with multiplicity two.   Thus we can choose the basis vectors $v_i ( i\geq 0)$ and 
$w_j (j\geq 0)$ such that
\begin{eqnarray*}
&& f \cdot v_i=v_{i+1};\\
&& h \cdot v_i= (\lambda-2 i ) v_i;\\
&& e \cdot v_i= i (\lambda- i +1) v_{i-1};\\
&& f \cdot w_i=w_{i+1};\\
&& h \cdot w_i= (-\lambda-2-2 i ) v_i;\\
&& e \cdot w_i= -i (\lambda +i+1) w_{i-1};\\
&& e \cdot w_0=v_{\lambda}.
\end{eqnarray*}
Notice that this indecomposable module $P(-\lambda-2)$ has a Verma submodule $M(\lambda)$,
as well as a Verma submodule $M(-\lambda-2)$. The quotient space $P(-\lambda-2)/M(\lambda)$ is isomorphic to $M(-\lambda-2)$ and 
the quotient space $P(-\lambda-2)/M(-\lambda-2)$ is isomorphic to $M^{\vee}(\lambda)$.

\section{Master symmetries of evolution equations}\label{sec3}
In this section we give the definitions of symmetries and master symmetries for scalar $1+1$-dimensional evolutionary partial differential equations
meanwhile we also fix some notations. These definitions can be easily extend to multicomponent and $2+1$-dimensional evolutionary equations.

Let $u$ be a scalar smooth function of the independent variables
$x$ and $t$. We consider an evolutionary differential equation of dependent variable $u$ of the form
\begin{equation}\label{eq}
u_t=K[u], 
\end{equation} 
where $[u]$ means that the smooth function $K$ depends on $u$ and $x$-derivatives of $u$ up to some finite order. 

All smooth functions depending on $x$, $t$, $u$ and $x$-derivatives of $u$ form a
differential ring $\A$ with total $x$-derivation
$$D_x=\pderiv{}{x}+\sum_{k=0}^\infty u_{(k+1)x}  \pderiv{}{u_{kx}},  $$ 
where $\ u_{kx}=\partial_x^k u$. We simply write as $u, u_x, u_{xx}, \cdots$ instead of $u_{0x}, u_{1x}, u_{2x}, \cdots$
when $k$ is small. The
highest order of $x$-derivative we refer to the order of a given function. For any
element $g\in \A$, we define an equivalence class (or a functional) $\int\! g$ by
saying that $g$ and $h$ are equivalent if and only if \(g-h\in \im D_x\). The space
of functionals, denoted by $\A'$, does not inherit the ring structure from $\A$.

A vector field (derivation) is said to be evolutionary if it commutes with
the operator $D_x$. In the scalar case, such vector filed is completely determined by a smooth 
function $P\in \A$. We call it the characteristic of the vector field 
$${\bf v}_P=\sum_{k=0}^\infty D_x^k(P)  \pderiv{}{u_{kx}}.$$
For any two evolutionary vector fields with characteristics $P$ and $Q$, we define a Lie bracket as follows
\begin{eqnarray}\label{liebra}
[P, \ Q]=Q_{\star} (P)-P_{\star}(Q), 
\end{eqnarray}
where $P_{\star}=\sum_j \frac{\partial P}{\partial u_{jx}} D_x^j$ is the Fr{\'e}chet derivative of $P$. The evolutionary
vector fields form a Lie algebra denoted by $\lieh$. We simply say $P\in \lieh$.
\begin{Def}\label{def1}
An evolutionary vector field with characteristic $P$ is a symmetry of system
(\ref{eq}) if and only if 
\begin{eqnarray}\label{symmetry}
 \pderiv{P}{t}+\ad_K P=\pderiv{P}{t}+[K,\ P]=0
\end{eqnarray}
This equation is said to be integrable if it possesses infinitely many higher
order symmetries. 
\end{Def}
Notice that if $P$ is not explicitly dependent on $t$, the symmetry condition will reduce to $\ad_K P=0$.
Since the right hand side of equation (\ref{eq}) does not explicitly depend on time $t$, we can formally find $t$-dependent symmetry $P$ as
(cf. \cite{mr86c:58158})
$$P=\exp (-t \ad_K) \phi$$
for any $\phi\in \A$ not explicitly depending on $t$. This expression
makes sense when it reduces to a finite sum.
In particular, if for some \(m\in\N\) one
has \(ad_K^m \phi=0\), then  
\begin{eqnarray*}
 P=\sum_{l=0}^{m-1} \frac{(-t)^l}{l!} \ad_K^l \phi
\end{eqnarray*}
is a symmetry of equation (\ref{eq}).
Moreover, the partial derivative of $P$ with respect to $t$ of order $r\in \N$,
that is, \(\frac{\partial^r P}{\partial t^r} \)
is also a symmetry of \(K\):
\[
\frac{\partial^{r+1} P}{\partial^{r+1} t}
+ad_K\frac{\partial^r P}{\partial t^r}=
\frac{\partial^r }{\partial t^r} (\pderiv{P}{t}+\ad_K P) =0.
 \]
In particular, this implies that if \(P\) is a  polynomial of degree \(p\)
in \(t\),
then \(\frac{\partial^p P}{\partial t^p}\) is a time independent
symmetry of \(K\).
Notice that it is essential for this argument to work that both \(K\) and \(\psi\)
have no explicit time-dependence.

This observation was made by Fuchssteiner \cite{mr86c:58158}, where he defined the concept of $K$-generators:
we say $\phi[x,u]$ is a $K$-generator of degree $m$ if $\ad_K^{m+1} \phi=0$ and $\ad_K^m \phi\neq 0$.
The well-known master symmetries of equation (\ref{eq}) are $K$-generator of degree $1$.

\begin{Def}
Let $a_0$ be a time independent symmetry of equation (\ref{eq}), that is, $[a_0, K]=0$,
We say a $t$-independent evolutionary vector field $\tau$ is a master symmetry for (\ref{eq}) if $a_n$ obtained by the recurrence
procedure 
$$a_{n+1}=[\tau, a_n], \qquad n=0,1,2, \cdots
$$
mutually commute, i.e.
$[a_i, a_j]=0$,
and there exists $k\in\mathbb{N}$ such that $a_k=K$.
\end{Def}
%Here we use the so-called $\tau$-scheme in Dorfman's book \cite{mr94j:58081} to define master symmetry. 
Obviously, it follows from the definition that the
existence of a master symmetry implies integrability.
It also implies that
\begin{enumerate}
 \item $\ad_K^2 \tau=0$, that is, $\tau$ is a $K$-generator of degree $1$;
 \item $t \ad_K \tau +\tau$ is a time dependent symmetry of equation (\ref{eq}).
\end{enumerate}
These two implications have been used in the construction of a master symmetry for a given equation. We refer to \cite{mr1465768} 
and its references for concrete examples.

The master symmetries can be used to construct the hierarchy of conserved densities for equation (\ref{eq}) if existing.
\begin{Def}\label{den}
We say $\rho\in \A'$ is a conserved density of equation (\ref{eq}) if $$\int \left( \frac{\partial \rho}{\partial t}+\rho_*(K)\right)=0.$$
\end{Def}
If $\rho$ does not explicitly depend on $t$, the definition becomes to
$$
0=\int \rho_*(K)=\int {\bf v}_K (\rho)=\int K \cdot \frac{\delta \rho}{\delta u},
$$
where $\frac{\delta \rho}{\delta u}=\sum_j (- D_x)^j \left(\frac{\partial \rho}{\partial u_{jx}}\right)$ is the variational derivative of $\rho$.
For any $Q\in\lieh$,  we have (it can be easily obtained by Leibnitz rule of
Lie derivatives \cite{mr94j:58081})
\begin{eqnarray}\label{taucon}
\int {\bf v}_Q ( \rho_*(K))=\int [Q, K] \frac{\delta \rho}{\delta u} + \int K
\frac{\delta \rho_*(Q)}{\delta u} .
\end{eqnarray}
\begin{Pro}\label{pro1}
Let $\rho$ be a conserved density and $Q$ be a symmetry for equation (\ref{eq}). If $\rho_*(Q)\not\equiv 0$, then it is also 
a conserved density of equation (\ref{eq}).
\end{Pro}
\begin{proof}
 According to Definition \ref{den}, we compute
 \begin{eqnarray*}
&&\quad\int \left( \frac{\partial \rho_*(Q)}{\partial t}+\left(\rho_*(Q)\right)_*(K)\right) 
=\int \left( \left(\frac{\partial \rho}{\partial t}\right)_* (Q)+\frac{\partial Q}{\partial t} \frac{\delta \rho}{\delta u}
+K \frac{\delta \rho_*(Q)}{\delta u}\right)\\
&&=\int {\bf v}_Q ( \frac{\partial \rho}{\partial t}+ \rho_*(K))+\int \left(\frac{\partial Q}{\partial t}+ [K, Q]\right) \frac{\delta \rho}{\delta u}=0 ,
 \end{eqnarray*}
where we used identity (\ref{taucon}) and Definition \ref{def1}. Thus we proved the statement.
\end{proof}
%This proposition can be used for non-evolutionary equations by consider $u_t=K$ as a symmetry flow of the equation. 
Taking $Q=\tau$ in (\ref{taucon}) and using Definition \ref{den}, we can formulate the following statement:
\begin{Pro}\label{pro2}
Let $\rho$ be a $t$-independent conserved density for equation (\ref{eq}). Assume that there is a $t$-independent evolutionary vector field $\tau$ such
 $\rho_*(\ad_{\tau} K)\equiv 0$. If $\rho_*(\tau)\not\equiv 0$, it is a conserved density of equation (\ref{eq}).
\end{Pro}
Combining Proposition \ref{pro1} and \ref{pro2},
If we know $\tau$ is a master symmetry, for any $a_i$ we can start from a conserved density $\rho_0$ of $u_t=a_i$
and define 
$$\rho_{1}\equiv\left\{\begin{array}{ll}
{\rho_{0}}_*(\ad_{\tau} a_i);& \mbox{when ${\rho_{0}}_*(\ad_{\tau} a_j)\not\equiv 0$}\\
{\rho_0}_*(\tau);& \mbox{when ${\rho_{0}}_*(\ad_{\tau} a_j)\equiv 0$}
\end{array}\right.$$
If $\rho_1\neq0$, we repeat the above procedure to define $\rho_2$. In this way we generate infinitely many conserved densities $\rho_i\neq 0$.
If there exists $k\in \mathbb{N}$ such that $\rho_k=0$,   equation $u_t=a_i$
likely possesses only finite number of conserved densities.

For Hamiltonian systems, master symmetries can also be used to generate
Hamiltonian vector fields \cite{mr94j:58081}. 
\begin{Pro}\label{pro3}
Let ${\cal H}$ be a Hamiltonian operator. Assume that $a_0$ is a Hamiltonian vector field, that is, there exists $f\in \A'$ such that 
$a_0={\cal H} \frac{\delta f_0}{\delta u}$. So is the master symmetry $\tau$. Then for all $n\in \mathbb{N}$, $a_n={\cal H} \frac{\delta f_n}{\delta u}$,
where $f_n\equiv{f_{n-1}}_*(\tau)$.
\end{Pro}
\begin{proof} Since $\tau$ is a Hamiltonian vector field, there exists $g\in \A'$ such that $\tau={\cal H}\frac{\delta g}{\delta u}$. 
Notice that 
\begin{eqnarray*}
 a_{1}=[\tau, a_{0}]= [{\cal H}\frac{\delta g}{\delta u}, {\cal H}\frac{\delta f_{0}}{\delta u}]
 ={\cal H} \frac{\delta \{g, f_{0}\}}{\delta u}={\cal H} \frac{\delta {f_{0}}_*(\tau)}{\delta u},
\end{eqnarray*}
where Poisson bracket $\{g, f_{0}\}$ is defined by Hamiltonian operator $\cal H$. So $a_1$ is a Hamiltonian vector field with Hamiltonian
$f_1\equiv{f_{0}}_*(\tau)$. By induction, we can prove the statement.
\end{proof}

Finally, we define the homogeneity for equation (\ref{eq}), which is a technical requirement for the approach to constructing master symmetries
presented in the next section.
\begin{Def}\label{def3}
We say equation (\ref{eq}) is homogeneous (or $\alpha$-homogeneous)
if there exist constant $\alpha$ and
$\kappa$ such that $[x u_x+\alpha u,\ K]=\kappa K$. 
For this fixed $\alpha$, we call $x u_x+\alpha u$ a scaling symmetry for the
equation.
\end{Def}
In next section, we show that we are able to construct $\liesl(2)$ around the
scaling symmetry. Thus, in other word, 
 when equation (\ref{eq}) is homogeneous, the right hand of the
equation is a basis for the weight space $M_{2 \kappa}$.
\section{Construction of master symmetries}\label{sec4}
In this section, we'll present a method to construct master symmetries based on the \(\liesl(2,\C)\) module. 
In fact, the concept of $K$-generators also reminds
one of the representation theory of \(\liesl(2,\C)\) \cite{ mr1874291}. In what
follows, we first present a natural $\liesl(2)$ for homogeneous equations and
further
construct an infinite dimensional module in the BGG category $\cO$.
\begin{Lem}\label{lem1}
Given a scaling $h=2 (x u_x +\alpha u)$, where $\alpha$ is constant, the
elements $e=u_x$ and $f=-(x^2 u_x+ 2 \alpha x u)$ form an $\liesl(2, \C)$.
\end{Lem}
\begin{proof}
It is easy to check the three given elements $h$, $e$ and $f$ satisfy the relations (\ref{sl2}).
\end{proof}
\begin{Rem}\label{rem1}
Note that for given $h$ and $e$ in the above lemma, the element $f$ is not unique. In fact they form $\liesl(2, \C)$ with any $f'=f+f_0$ when $f_0$ satisfies
$[e, f_0]=0$ and $[h, f_0]=-2 f_0$. For example,  in (\ref{sl2b}), instead of $-(x^2 u_x+2x u)$ we take $f=-(x^2 u_x +2 x u+2)$.
\end{Rem}
For convenience, we sometimes denote this $\liesl(2, \C)$ as $\liesl_d (2)$
indicating its three elements lie in the diagonal direction (see Diagram
\ref{dia1}).
Notice that $\ad_e w=\ad_{u_x} w=-\frac{\partial w}{\partial x}$ for $w\in \lieh$. So we can compute the inverse action of $\ad_e$.
 This means that for any $g \in \lieh$ we are able to find $w\in \lieh$ such that $\ad_e w=g$.
\begin{Lem}\label{lem2}
Let $w=\frac{1}{3} \left(x^3 u_x+3 \alpha x^2 u\right)$. We have $\ad_e w=f$. Moreover,
\begin{eqnarray}\label{adfn}
\ad_f^n w=\frac{n!}{3} \left(x^{n+3} u_x +(n+3) \alpha x^{n+2} u\right),  \quad n=0, 1, 2,\cdots 
\end{eqnarray}
and $\ad_h \ad_f^n w=-2(n+2)\ad_f^n w$.
\end{Lem}
\begin{proof}
By direct computation we have $\ad_e w=-\frac{\partial w}{\partial x}=f$ and $\ad_h \ad_f^n w=-2(n+2)\ad_f^n w$.

We now prove (\ref{adfn}) by induction. It is obviously true for $n=0$. Assume that the formula is valid for $n$. Then for $n+1$, we have
\begin{eqnarray*}
&&\ad_f^{n+1} w =[f,\ \ad_f^n w]=[-(x^2 u_x+ 2 \alpha x u),\ \frac{n!}{3} \left(x^{n+3} u_x +(n+3) \alpha x^{n+2} u\right)]\\
&&\quad = \frac{n!}{3}\left(x^2 D_x \left(x^{n+3} u_x +(n+3) \alpha x^{n+2} u\right)+2 \alpha x \left(x^{n+3} u_x +(n+3) \alpha x^{n+2} u\right)\right)
\\&&\qquad
-\frac{n!}{3} \left(x^{n+3} D_x(x^2 u_x+ 2 \alpha x u) +(n+3) \alpha x^{n+2} (x^2 u_x+ 2 \alpha x u)\right)
\\&&\quad
=\frac{n!}{3}\left((n+1)x^{n+4} u_x+(n+1) (n+4) \alpha x^{n+3} u \right)
=\frac{(n+1)!}{3}\left(x^{n+4} u_x+ (n+4) \alpha x^{n+3} u \right),
\end{eqnarray*}
which is the right hand side of (\ref{adfn}) for $n+1$. So we proved the statement.
\end{proof}
\begin{Rem}\label{rem2}
In Lemma \ref{lem2}, the element $w$ satisfying $\ad_e w=f$ and $\ad_h w=-4 w$ is not unique. We can take $w+w_0$ instead of $w$ as long as
$[e, w_0]=0$ and $[h, w_0]=-4w_0$.
\end{Rem}
From the above lemma, we know that all $\ad_f^n w$ are independent over $\C$. They generate an infinite dimensional space over $\C$. Together with Lemma \ref{lem1},
we get the following $\liesl_d(2)$ module:
\begin{eqnarray}\label{quoti}
 e \pile{\rTo^{f} \\ \lTo_{e}}   h \pile{\rTo^{f} \\ \lTo_{e}}  f \lTo_{e}  w \pile{\rTo^{f} \\ \lTo_{e}}  \ad_f w \pile{\rTo^{f} \\ \lTo_{e}} 
  \ad_f^2 w \cdots\cdots,  
\end{eqnarray}
where we used the following notation: 
\begin{Not}\label{not1}
In a diagram, \(f\stackrel{N}{\longrightarrow}g\) means $\ad_N f=\gamma g$, where \(\gamma\) is a nonzero constant. 
\end{Not}
This module is the quotient space $P(-4)/M(-4)$.  
This module and its freedom stated in both Remark \ref{rem1} and
Remark \ref{rem2} play an important role in constructing master symmetries 
for homogeneous evolution equations and computing  their time-dependent symmetries.

For a homogeneous evolutionary equation (\ref{eq}), no matter it is integrable or not, there exist the $\liesl_d(2)$ and further its infinite dimensional
module shown (\ref{quoti}). In the following theorem, we claim that $\ad_f K$ (or $\ad_{f'} K$ according to Remark \ref{rem1}) is a master symmetry for integrable equation $u_t=K[u]$ under certain conditions.

\begin{The}\label{thm1}
Take $a_0=e$ and
$a_1\in \lieh$ satisfying 
\begin{eqnarray*}
[a_0, a_1]=0 \quad \mbox{and} \quad [h, a_1]=[2 (x u_x +\alpha u), a_1]=\lambda
\ a_1, \  \mbox{where} \ \lambda>2 
\end{eqnarray*}
for a constant  $\alpha$. Let $f=-(x^2 u_x+ 2 \alpha x u)$.
Define $\tau=-\frac{1}{\lambda}\ad_f a_1$ and $a_{n+1}=[\tau, a_n]$ for $n=0, 1, 2, \cdots.$ Suppose that there exists an integer $s\geq 2$ such that
\begin{enumerate}
 \item[{\rm (i)}] $\qquad \qquad
 [a_{n-1}, a_n]=0, \quad n=1,2,\cdots, s.$
 \item[{\rm (ii)}]  For any $n\in \N$, from $[[a_n, a_{n+1}], a_{k}]=0$ for $0\leq k \leq s-1$ it follows that there exists $m\leq 2 s\in \N$ such that 
$[a_n, a_{n+1}]=\sum_{j=0}^{m} \gamma_j a_j$, where $\gamma_j$ are constant.
\end{enumerate}
Then all $a_n$ mutually commute, that is, $[a_i, a_j]=0,\ \ i,j=0, 1, 2, \cdots.$
\end{The}
Before we give the proof of this theorem, we first check the definition for $a_1=[\tau, a_0]$ is consistent with the definition of $\tau$. Indeed,
\begin{eqnarray*}
 [\tau, a_0]=[\tau, e]=-\frac{1}{\lambda} [[f, a_1], e]=\frac{1}{\lambda}\left([[e, f],a_1]+[[a_1, a_0],f]\right)=\frac{1}{\lambda}[h,a_1]=a_1.
\end{eqnarray*}
The proof of this theorem is very similar to the one for $\tau$-scheme in Dorfman's book \cite{mr94j:58081}.
However, to make the paper readable, we write it in details and begin with two lemmas.
\begin{Lem}\label{lem30}
Assume that $a_n$ are defined as in Theorem \ref{thm1}. Then $[h, a_n]=(n \lambda -2n +2) a_n$ for $n=0, 1, 2, \cdots$.
\end{Lem}
\begin{proof}
We prove the statement  by induction. When $n=0$, it is valid by assumption.
Assume it works for $n=l$.
When $n=l+1$, we get by the Jacobi identity
\begin{eqnarray*}
[h, a_{n+1}]&=&[h, [\tau, a_n]]=-[\tau, [a_n, h]]-[a_n, [h, \tau]]\\
&=&(n \lambda -2n +2) [\tau, a_n]-(\lambda-2) [a_n, \tau] =( (n+1)\lambda-2n) a_{n+1}.
\end{eqnarray*}
Here we used the fact that $\tau=\frac{1}{\lambda} [a_1, f]$ has the weight $\lambda -2$.
\end{proof}

\begin{Lem}\label{lem3}
Assume that $a_n$ are defined as in Theorem \ref{thm1} and $[a_{n-1}, a_n]=0$ for $n\leq l$. Then $[a_m, a_n]=0$ for all
$m+n\leq 2 l$ .
\end{Lem}
\begin{proof}
We again prove it by induction.
When $l=0$, this statement is trivially valid. Suppose that the statement is
valid for some $l$. We now prove this statement for $l+1$.

For all pairs of $m$ and $n$, we assume $m\leq n$ without losing any generality.
There are three cases: (i). $m+n\leq 2l$; (ii). $m+n=2l+1$ and (iii).
$m+n=2l+2$.
When $m+n\leq 2l$, this is covered by the induction assumption and thus there is
no need to prove.
For other two cases, without losing generality, we assume $m<n$. By the Jacobi identity, we have
\begin{eqnarray}
 [a_m, a_n]&=&[a_m, [\tau, a_{n-1}]]=-[\tau, [a_{n-1},
a_m]]-[a_{n-1},[a_m,\tau]]\nonumber\\
 &=&-[\tau, [a_{n-1}, a_m]]+[a_{n-1},a_{m+1}].\label{jacobi}
\end{eqnarray}
For case (ii), we know $[a_{n-1}, a_m]=0$ since $m+n-1\leq 2l$. So we get
\begin{eqnarray*}
 [a_m, a_n]= [a_{n-1},a_{m+1}]=-[a_{m+1}, a_{n-1}].
\end{eqnarray*}
By repeating applying this identity, we obtain $[a_m, a_n]=
(-1)^{\frac{n-m-1}{2}}[a_{l},a_{l+1}]$, which vanishes as an assumption.

\noindent
For case (iii), we have $[a_{n-1}, a_m]=0$ since $m+n-1\leq 2l+1$. In the same
way as we did for case (ii), we get 
$$[a_n, a_m]=[a_{l+1},a_{l+1}]=0$$ 
and thus we complete the proof.
\end{proof}
We are now ready to complete the proof of Theorem \ref{thm1}.
\begin{proof} To prove Theorem \ref{thm1}, we only need to show that $[a_{n-1}, a_n]=0$ for all $n\in \N$. We prove it by induction.

From condition (i), this identity is valid for $n=1,2,\cdots, s.$
Assume that it is valid for all $n\leq l$ and $l\geq s$. For $n=l+1$, we have by the Jacobi
identity
\begin{eqnarray*}
 [[a_l, a_{l+1}], a_{k}]=-[[a_{l+1},a_{k}],a_l]-[[a_{k},a_l],a_{l+1}]
\end{eqnarray*}
It follows from Lemma \ref{lem3} that $[a_{l+1},a_{k}]=[a_{k},a_l]=0$ since
$l+1+k\leq 2 l$ for $0\leq k\leq s-1$. This leads to 
\begin{eqnarray*}
 [[a_l, a_{l+1}], a_{k}]=0, \qquad 0\leq k\leq s-1.
\end{eqnarray*}
This implies that $[a_l, a_{l+1}]=\sum_{j=0}^{m} \gamma_j a_j$, where $m\leq
2s$ from condition (ii).  We use $\ad_h$ acting on both sides of it and it
follows that
\begin{eqnarray*}
&& 0=[h, [a_l, a_{l+1}]-\sum_{j=0}^{m} \gamma_j [h, a_j]\\
&&\quad =[[h, a_l], a_{l+1}]+[a_l, [h, a_{l+1}]]-\sum_{j=0}^{m} \gamma_j (j \lambda-2 j+2)\\
&&\quad =\sum_{j=0}^{m} \left(\lambda +(2l-j)(\lambda -2)\right) \gamma_j a_j .
\end{eqnarray*}
Since $l\geq s$ and $\lambda>2$, we have $\lambda +(2l-j)(\lambda
-2)>2$ for all $0\leq j\leq m \leq 2 s$. Thus all $\gamma_j=0$
since $a_j$ are independent over $\C$. This leads to $[a_l, a_{l+1}]=0$ and we complete the proof.
\end{proof}
Following from this theorem we know $[e, a_n]=0$ and $a_n$ with weight $(\lambda -2)n+2$, where $\lambda>2$.
Therefore, starting from the highest weight vector $a_n$ the space generated by $\ad_f^j a_n,$ where $j=0, 1, 2, \cdots$, is an $\liesl(2, \C)=\{e, h, f\}$
module in the BGG category $\cO$. For a fixed $n$, if there exists $k\in \mathbb{N}$ such that $\ad_f^{k+1} a_n=0$, then
$\{\ad_f^j a_n, j=0, 1, 2, \cdots, k \}$
is a finite dimensional simple module $L((\lambda -2)n+2)=L(k)$. In this case,
we can find $\bar w$ such that $\ad_e \bar w=\ad_f^k a_n$ as in Lemma
\ref{lem2}. Thus we have 
\begin{eqnarray*}
\{\ad_f^j a_n, \ad_f^l \bar w, j=0, 1, 2, \cdots, k; l=0,1, 2 , \cdots \}
\end{eqnarray*}
is a $\liesl(2)$ module $P(-(\lambda -2)n-4)/M(-(\lambda -2)n-4)$.

The immediate application of Theorem \ref{thm1} lies in constructing a master symmetry for equation (\ref{eq}):
\begin{The}\label{thm2}
For homogeneous evolutionary equation $u_t=K[u]$ satisfying 
$$[u_x, K]=0 \quad \mbox{and} \quad [x u_x+\alpha u, K]=\kappa K, \quad \kappa>1$$ 
for a certain constant $\alpha$. Define 
$$\tau=\frac{1}{2 \kappa} [x^2 u_x+ 2 \alpha x u, K]$$
and $a_{n+1}=[\tau , a_n]$ with  $a_0=e=u_x$. Assume that
\begin{enumerate}
 \item[{\rm (i)}] $\qquad \qquad
 [[\tau, K], K]=0; $
 \item[{\rm (ii)}]  If there is a Lie subalgebra $\lieh'$ such that $a_n\in \lieh'$ for all $n=0, 1, 2, \cdots$. Moreover, 
 for any $P, Q \in \lieh'$ satisfying $[P, e]=[Q, e]=[P, K]=[Q, K]=0$, it
follows that $[P, Q]=0$.
%For any $n\in \mathbb{N}$ satisfying $[[a_n,a_{n+1}], K]=0$,  it follows that
%$[a_n, a_{n+1}]=0$.
\end{enumerate}
Then $\tau$ is a master symmetry of equation $u_t=K[u]$, that is, all $a_n$ mutually commute.
\end{The}
\begin{proof} Following the definition, we have  $a_1=[\tau, a_0]=\frac{1}{2 \kappa} [[-f, K], e]=K$. 
Thus the given conditions satisfy Theorem \ref{thm1} for $s=2$.  This leads to $\tau$ is a master symmetry of the equation.
\end{proof}
This theorem provide us a method to construct master symmetries of homogeneous
evolutionary equations. Moreover, starting from each $a_n$ obtained from this
theorem, we can construct an infinite dimensional $\liesl(2)$ module.  We
shall refer it as the $\cO$-scheme and present it in the following
diagram:
\begin{figure}[hz]
\captionsetup{name=Diagram}
\centering
\begin{diagram}
  u_x        &         &           &          &            &             &        &           &\\
 \dTo_{\tau}   & \rdTo^f &           &          &            &             &        &           &\\
 K        &         &  h        &          &            &             &         &           &\\
\dTo_{\tau} & \rdTo^f &           & \rdTo^f  &            &             &         &           &\\
a_2         &         & \tau      &          & f          &             &         &            &\\
\dTo_{\tau} & \rdTo^f &           & \rdTo^f  &            &    \luTo^{e} &        &            &\\
a_3         &         & \ad_f a_2 &         & \ad_f^2 a_1  &             &    w    &           &\\
\dTo_{\tau} & \rdTo^f &           & \rdTo^f  &            &    \rdTo^f   &        &  \rdTo^f   &\\
            &         &           &          &            &              &        &             &
\end{diagram}
\caption{ The $\cO$-scheme for
$u_t=K[u]$}\label{dia1}
\end{figure}

We are going to show that the elements in the horizontal directions
are related to the $t$-coefficients of $t$-dependent symmetries for equation
$u_t=K$. First we make a few remarks on Theorem \ref{thm2}.
\begin{Rem}\label{rem3}
\begin{enumerate}
\item The condition $\kappa>1$ is also trivially valid for all scalar homogeneous equations with linear terms of order greater than $1$.
\item Often we ignore the constant factor in $\tau$ (since it doesn't affect the commutativity of $a_j$) and simply say  $\tau=[x^2 u_x+ 2 \alpha x u, K]$.
\item  Note that condition {\rm (ii)} in Theorem \ref{thm1} is a technical condition, which only used once in the proof. The corresponding condition in Theorem \ref{thm2}
is a special case. 

It is valid  for scalar homogeneous polynomial equations with
linear terms. From the construction, we know all $a_n$ belong to the Lie subalgebra of homogeneous polynomial vector fields with nonzero linear terms.
For such equations, all $x$, $t$-independent polynomial
symmetries start with linear terms \cite{MR99g:35058}. The Lie bracket of two
symmetries, which is also a symmetry, has no linear terms. So it must be zero.

This argument can be generalised to multi-component homogeneous polynomial integrable systems with nonzero linear terms \cite{mr99i:35005}. 
The examples in next section all belong to this class.
\end{enumerate}
\end{Rem}

\begin{Lem}\label{lem50}
Under the conditions of Theorem \ref{thm2}, for any $n\in \N$ we have 
$$\ad_{e} \ad_K^{n-1} \ad_f^{n-1} a_j=\ad_{K}^{n} \ad_f^{n-1} a_j=0,  \quad j=0, 1, 2, \cdots. $$
\end{Lem}
\begin{proof}
We prove the statement by induction. For $n=1$, it is the result of Theorem \ref{thm2}. Assume it has been proved for all value less than $n$.
We will prove the same statement for $n$. First notice that for all $m\in \N$
we have
\begin{eqnarray}\label{change}
 \ad_K^m \ad_f=\sum_{i=0}^m \binom{m}{i} \ad_{(\ad_K^i f)} \ad_K^{m-i} =\sum_{i=0}^2 \binom{m}{i} \ad_{(\ad_K^i f)} \ad_K^{m-i}
\end{eqnarray}
since $\ad_K^i f=0$ for $i\geq 3$. Thus
\begin{eqnarray*}
\ad_K^n \ad_f^{n-1} a_j=\sum_{i=0}^2  \binom{n}{i} \ad_{(\ad_K^i f)} \ad_K^{n-i} \ad_f^{n-2} a_j
=\binom{n}{2} \ad_{(\ad_K^2 f)} \ad_K^{n-2} \ad_f^{n-2} a_j.
\end{eqnarray*}
Here we used the induction assumption $\ad_{K}^{n-1} \ad_f^{n-2} a_j=0$.
Moreover, we know both $\ad_K^2 f$ and $\ad_K^{n-2} \ad_f^{n-2} a_j$
commute with $K$ and $e=u_x$. Thus  we have $\ad_{(\ad_K^2 f)} \ad_K^{n-2} \ad_f^{n-2}
a_j=0$ by the condition (ii) in Theorem \ref{thm2}.

We now show that $\ad_{e} \ad_K^{n-1} \ad_f^{n-1} a_j=0$.
\begin{eqnarray*}
&&\ad_{e} \ad_K^{n-1} \ad_f^{n-1} a_j= \ad_K^{n-1} \ad_e \ad_f^{n-1} a_j
 = \ad_K^{n-1} \ad_h \ad_f^{n-2} a_j+ \ad_K^{n-1} ad_f \ad_e \ad_f^{n-2} a_j\\
&&\quad = \ad_K^{n-1} ad_f^2 \ad_e \ad_f^{n-3} a_j=\ad_K^{n-1} ad_f^{n-1} \ad_e  a_j=0.
\end{eqnarray*}
Here we used $\ad_K^{n-1} \ad_h \ad_f^{n-2} a_j=\ad_K^{n-2} \ad_h \ad_K \ad_f^{n-2} a_j=0$ and
the induction assumption $\ad_{K}^{n-1} \ad_f^{n-2} a_j=0$.
 Thus we complete the proof.
\end{proof}

Notice that $\ad_f^{n} a_j$ may be zero for fixed $j$ when $n$ is bigger. For a given example, we can determine the values of $n$ for fixed $j$.
If $\ad_f^{n} a_j\neq 0$, they are $K$-generators of degree $n$ following from Lemma \ref{lem50}. So the $t$-dependent
symmetries for equation (\ref{eq}) are
\begin{eqnarray*}
 \exp(-t \ad_K) \left( \ad_f^{n} a_j \right).
\end{eqnarray*}

Note that there is a unique independent symmetry given the degree in \(t\)
and the order of polynomial. E.g., \(\frac{\partial}{\partial t}
\left(\exp (-t \ad_K) \ad_f^2 a_1\right)\)
is dependent with \(\exp (-t \ad_K) \ad_f a_2\) in the sense that
the difference between these two can be expressed as the sum of
low order symmetries.
Therefore we only need the elements on the contour of the diagram 
to generate the independent symmetries. 
The following theorem will answer when $\ad_f^j w$ for $j=0,1,2, \cdots,$ are $K$-generators.
\begin{Lem}\label{lem60}
Assume that there exists $r\in \N$ such that $\ad_K^r w=0$ and $\ad_e\ad_K^{r-1} w=0$. Under the conditions
of Theorem \ref{thm2}, we have
$$\ad_e \ad_{K}^{n+r-1} \ad_f^{n} w=\ad_{K}^{n+r} \ad_f^{n} w=0, \quad  n=0, 1, 2, \cdots. $$
\end{Lem}
\begin{proof} The proof of this theorem is similar to Lemma \ref{lem50}. 
We prove the statement by induction. For $n=0$, it is assumption of the
statement. Assume it has been proved for all value less than $n$.
We will prove the same statement for $n$. Using (\ref{change}), we get
\begin{eqnarray*}
\ad_K^{n+r} \ad_f^{n} w=\sum_{i=0}^2  \binom{n+r}{i} \ad_{(\ad_K^i f)}
\ad_K^{n+r-i} \ad_f^{n-1} w
=\binom{n}{2} \ad_{(\ad_K^2 f)} \ad_K^{n+r-2} \ad_f^{n-1} w.
\end{eqnarray*}
Here we used the induction assumption $\ad_{K}^{n+r-1} \ad_f^{n-1} w=0$.
Since both $\ad_K^{n+r-2} \ad_f^{n-1} w$  and $\ad_K^2 f$ commute with $K$ and $e=u_x$,
so $\ad_{(\ad_K^2 f)} \ad_K^{n+r-2} \ad_f^{n-1} w=0$ by the condition (ii) in
Theorem \ref{thm2}.

We now show that $\ad_e \ad_{K}^{n+r-1} \ad_f^{n} w=0$.
\begin{eqnarray*}
&&\ad_e \ad_{K}^{n+r-1} \ad_f^{n} w= \ad_{K}^{n+r-1} \ad_e \ad_f^{n} w
 = \ad_{K}^{n+r-1} \ad_h \ad_f^{n-1} w+\ad_{K}^{n+r-1}\ad_f  \ad_e \ad_f^{n-1} w\\
&&\quad = \ad_{K}^{n+r-1} \ad_f^{n} \ad_e w=\ad_{K}^{n+r-1} \ad_f^{n} f=0.
\end{eqnarray*}
Thus we proved the statement.
\end{proof}

Following from Diagram \ref{dia1} and Lemma \ref{lem50}, instead of direct searching for the integer $r$ such that $\ad_K^r w=0$  we
check whether there exists a constant $\gamma\neq 0$ such that 
$$\ad_K w=\gamma \ad_f^2 K=\gamma \ad_f^2 a_1.$$
If this is the case, then $w$ is a $K$-generator of degree $3$, i.e., $\ad_K^4 w=0$. Moreover, $\gamma$ can be determine by the weight of $K$.
\begin{Lem}
Let $f$ $h$ and $w$ be defined as in Lemma \ref{lem1} and \ref{lem2}. Assume that $[h, K]=2 \kappa K$, where $\kappa>1$.
If $\gamma \ad_K w= \ad_f^2 K$, then $\gamma=2 -4 \kappa$.
\end{Lem}
\begin{proof}
 Notice that $\frac{\partial w}{\partial x}=-f$ and $\frac{\partial f}{\partial x}=-h$. We differentiate both sides of $\gamma \ad_K w= \ad_f^2 K$ with respect to $x$
 and it follows that
 \begin{eqnarray*}
\gamma \ad_K (-f)=[-h, [f,K]]+[f,[-h,K]]=(-(2 \kappa-2)-2 \kappa)[f, K] ,
 \end{eqnarray*}
which implies that $\gamma=2-4 \kappa$ since $[f, K]\neq 0$.
\end{proof}
Therefore, we have the following result on $t$-dependent symmetries for equation (\ref{eq}):
\begin{The}\label{thm3}
Assume that the homogeneous equation $u_t=K[u]$ satisfies the conditions in Theorem \ref{thm2}. 
Let the evolutionary vector field $w$ satisfy $\ad_e w=f$ and $\ad_K w=\frac{1}{2-4 \kappa} \ad_f^2 K$. Then
$$\frac{\partial^j }{\partial t^j} \exp(-t \ad_K) \left( \ad_f^{n} w \right),
\quad \mbox{where}\ n=0, 1, 2, \cdots, \mbox{and} \ j=0, 1, 2, \cdots, n+2$$
are $t$-dependent symmetries of the equation.
 \end{The}
\begin{proof}
Under the assumption we have shown that $w$ is a $K$-generator of degree $3$. From  Lemma \ref{lem2} we know that $\ad_f^j w\neq 0$ for all $j\in \N$
and we have $\ad_K^{n+4} \ad_f^n w=0$ for all $n\in \N$ following from Lemma \ref{lem60}. 
Therefore, symmetry $\exp(-t \ad_K) \left( \ad_f^{n} w \right)$ is a polynomial in $t$ of degree $n+3$.
Since the given equation does not depend $x$ and $t$ explicitly, the $t$ partial derivatives of a symmetry is still a symmetry.
\end{proof}
We remark for some equations such as the Burgers equation studied next section, this theorem does not give us all time dependent symmetries.

\section{Examples and applications}\label{sec5}
In this section, we'll demonstrate the approach in Section \ref{sec4} construct a master symmetries for both multi-component
and multi-dimensional evolutionary equations. We apply Theorem \ref{thm2} to
some
new equations. We also include
some well-known examples, e.g., the Burgers equation in order to compare our method to the existing methods. 
\subsection{$(1+1)$-dimensional differential equations}
In this section, we use the approach described in Section \ref{sec4} to construct known results on master symmetries for a few typical examples
such as the Burgers equation (one component homogeneous equation), the Landau-Lifshitz equation (two-component non-homogeneous equation) and Benjamin-Ono equation
(integro-differential equation). Meanwhile we also compare our
approach to the existing methods.
\subsubsection{The Burgers equation}\label{sec41}
Consider the Burgers equation 
\begin{eqnarray}\label{bug}
u_t=K=\uu{xx}+u \uu{x} , 
\end{eqnarray}
which is homogeneous for $\alpha=1$ and we have
\begin{eqnarray*}
 [x u_x+u, K]=2 K.
\end{eqnarray*}
According to Theorem \ref{thm2}, its master symmetry is
$$\tau=\frac{1}{4} [x^2 u_x+ 2 x u, K]=x(u_{xx}+u u_x)+\frac{3}{2}u_x +\frac{1}{2} u^2$$
since 
\begin{eqnarray*}
 a_2=[\tau, K]=2 u_{xxx}+3u u_{xx}+3 u_x^2+\frac{3}{2} u^2 u_{x}
\end{eqnarray*}
and $[[\tau, K], K]=0$, and condition (ii) is trivially satisfied following from Remark \ref{rem3}.

We now discuss time dependent symmetries for equation (\ref{bug}). Let $w=\frac{1}{3}x^3 u_x+x^2 u$. Notice that
\begin{eqnarray*}
&&\ad_K w=-2 (x^2 u u_x+x^2 u_{xx}+3 x u_x +x u^2 +u)=-\frac{1}{6} \ad_f^2 K.
\end{eqnarray*}
According to Theorem \ref{thm3}, we can write down the time dependent symmetries for equation (\ref{bug}). For example,
\begin{eqnarray*}
 &&S=\exp(-t \ad_K)  w =\frac{4}{3} t^3 \left(2 u_{4x}+4 u u_{3x}+10 u_x u_{xx}+3 u^2 u_{xx}+6 u u_x^2+u^3 u_x \right)\\
 &&\qquad +t^2 \left(x(4 u_{3x}+6 u u_{xx}+6 u_x^2+3 u^2 u_x)+8 u_{xx}+10 u u_x+u^3 \right)\\
 &&\qquad+2 t \left(x^2(u_{xx}+u u_x)+x(3 u_x+u^2)+u \right)+\frac{1}{3}x^3 u_x+x^2 u
\end{eqnarray*}
and $\frac{\partial^j }{\partial t^j} S$ for $j=0,1,2$ are time dependent
symmetries. Comparing to the list of $t$-dependent symmetries $S_n$,
$n=1, 2, 3, 4$, of the Burgers equation given in Introduction, we have
$$
 \frac{\partial S_4}{\partial t}=\frac{3}{2} S+3 S_2.
$$
Following from  Lemma \ref{lem50}, both $\ad_{f}a_2$ and $\ad_{f}^2 a_1$ are $K$-generators. However, they do not generate any new symmetries. Indeed, we have
\begin{eqnarray*}
&&\exp(-t \ad_K) \left( \ad_{f}a_2\right)=-\frac{3}{2} \frac{\partial^2 S}{\partial t^2}\quad \mbox{and} \quad 
\exp(-t \ad_K) \left( \ad_{f}^2 a_1 \right)=6 \frac{\partial S}{\partial t}.
\end{eqnarray*}

The master symmetries and time dependent symmetries for the Burgers equation are known in the literature (see, for example, \cite{mr86c:58158, mr2002b:37100}).
We now discuss how to explain other approaches under the frame of the
$\liesl(2)$
representations.

In paper \cite{mr1874291}, the authors constructed two different \(\liesl(2, \C)\) around the scaling \(h=x\uu{1}+u\)
with generators $e_i, h_i$ and $f_i$, $i=1,2$ as follows:
\begin{eqnarray*}
&&\liesl_h(2, \C): e_1=K=\uu{xx}+u \uu{x}, \ f_1=-x,\ h_1=\frac{h}{2}= x u_x+u;\label{sl21}\\
&&\liesl_v(2, \C): e_2=1, \ f_2=-2\tau=-2x (u_{xx}+u u_x)-3 u_x-u^2,\ h_2=-h=-2(x u_x+u).\label{sl22}
\end{eqnarray*}
Here we can take any $f_2+\beta u_x$, where $\beta$ is constant, as $f_2$. To be
consistent with $\tau$ defined previously, we choose the one in $\liesl_v(2,
\C)$.
Together with $\liesl_d(2)$ given in Lemma \ref{lem1}, we present them in
Diagram 2. 

\begin{figure}[hz]
\captionsetup{name=Diagram}
\centering
\begin{diagram}
\liesl_d(2)  &        &        &    \liesl_v(2)        &          &            &             &        &           &\\
& a_0        &         & e_2          &          &            &             &        &           &\\
& \dTo   & \rdTo^f &   \dTo^{\tau}        &          &            &             &        &           &\\
 \liesl_h(2)\ & a_1        &  \rTo^{f_1}       &  h        &  \rTo^{f_1}         &  f_1          &             &         &           &\\
&\dTo_{\tau} & \rdTo^f &   \dTo^{\tau}        & \rdTo^f  &            &             &         &           &\\
& a_2         &         & \tau      &          & f          &             &         &            &
\end{diagram}
\caption{Three $\liesl(2, \C)$ for the Burgers equation}\label{dia2}
\end{figure}

There three $\liesl(2)$ are the foundation of different approaches to
constructing master symmetries.
Dorfman's $\tau$-scheme \cite{mr94j:58081}
is based on $\liesl_v(2)$. Since $[h_2, a_0]=-2 a_0$ and $[e_2, a_0]=0$, all symmetries $a_j$ are the basis for Verma module $M(-2)$ of  $\liesl_v(2)$.
In the $\tau$-scheme, the element $\tau$ is already known. So it can only be used to prove that the given $\tau$ is indeed a master symmetry.

Since $[a_n, K]=0$ and $[h_1,a_n]=(n+1) a_n$,  the space spanned by 
$$\{\ad_{f_1}^i a_n,\quad  i=0, 1, \cdots, n+1\}$$
is an $(n+2)$--dimensional irreducible representation of $\liesl_h(2)$. This is the theory behind Fuchssteiner's approach in \cite{mr86c:58158}.
According to him, a master symmetry for equation $u_t=K[u]$ is $[a_2, f_1]$. Indeed, for the Burgers equation, 
\begin{eqnarray*}
 [a_2, f_1]=3x(u_{xx}+u u_x)+6u_x +\frac{3}{2} u^2=3 \tau+\frac{3}{2} u_x
\end{eqnarray*}
is a master symmetries since $[u_x, a_n]=0$.

The method proposed in this paper is based on $\liesl_d(2)$.  In fact, this method was first used in \cite{mr1874291} for Ibragimov-Shabat equation
\begin{eqnarray}\label{IS}
 u_t=u_{xxx}+3 u^2 u_{xx}+9 u u_x^2+3 u^4 u_x .
\end{eqnarray}
and the Kadomtsev--Petviashvili equation (\ref{kp}) (we revisit it in section \ref{seckp}). Later in \cite{ffokas02}, 
Finkel and Fokas adapted it (without specifying $\liesl(2)$) to deal with
nonlocal master symmetries
such as the Sawada-Kotera equation. 
The advantage of using $\liesl_d(2)$ instead of other $\liesl(2)$ is that it is a free structure for homogeneous evolutionary equations. 
There are no need to search for $e_2$ and $f_1$. Indeed, for some integrable equations, they do not exist, for instance, the Ibragimov-Shabat equation.
Since we do not reach out for $f_1$, consequently, we do not find all $t$-dependent symmetries.

Finally, we comment that using the quotient space $P(-4)/M(-4)$,
see (\ref{quoti}) instead of only $\liesl_d(2)$ enables us to find more time
dependent symmetries (still may not be all).
This leads to some $t$-dependent symmetries
for equation (\ref{IS}) are missing in \cite{mr1874291}, which was addressed in
\cite{mr1952867}.
\subsubsection{The Landau-Lifshitz equation}
Consider the Landau-Lifshitz equation
\begin{eqnarray}\label{lleq}
 \s_t=K=\s\wedge \s_{xx}+\s\wedge J\s ,
\end{eqnarray}
where $\s$ is a vector function of $x$ and $t$ in $\R^3$ with $|\s|=1$, $\wedge$ denotes the vector product
and $J$ is $3\times 3$  constant diagonal matrix. It is not homogeneous for only dependent variable $\s$. However,  if we add equation 
$J_t=0$, the system is homogeneous satisfying
\begin{eqnarray*}
 \big{[}\left( \begin{array}{c} x  \s_x +\alpha \s\\xJ_x+2J \end{array}\right), \left( \begin{array}{c}K\\0 \end{array}\right)\big{]}= (2+\alpha)
 \left( \begin{array}{c}K\\0 \end{array}\right)
\end{eqnarray*}
Following from Theorem \ref{thm2}, we define $\tau$ as the first component (the second component can be dropped since it is always zero) of 
\begin{eqnarray*}
\frac{1}{2(2+\alpha)} \big{[}\left( \begin{array}{c} x^2 \s_x+2x \alpha \s\\x^2 J_x+4 xJ \end{array}\right), \left( \begin{array}{c}K\\0 \end{array}\right)\big{]}
=\left( \begin{array}{c} x (\s\wedge \s_{xx}+\s\wedge J\s)+\frac{1+2 \alpha}{2+\alpha}\s\wedge \s_x\\0 \end{array}\right).
\end{eqnarray*}
We now check Condition (i) in Theorem \ref{thm2}. First we compute
\begin{eqnarray*}
 &&a_2=[\tau, K]=(2+\frac{1+2 \alpha}{2+\alpha})) \s\wedge (\s_x\wedge \s_{xx})+2 \s\wedge (\s\wedge\s_{xxx})+2\s\wedge (\s_x\wedge J\s) \\
&&\qquad +(2-\frac{1+2 \alpha}{2+\alpha})\s\wedge (\s\wedge J\s_x)
 +\frac{1+2 \alpha}{2+\alpha} \s \wedge J(\s\wedge \s_x)
\end{eqnarray*}
using $|\s|=1$ and $\s \cdot \s_x=0$. Direct computation shows that it commutes with $K$ when $\alpha=1$.  
Condition (ii) was explicitly proved in \cite{Fuch84}. 
Thus we have the master symmetry
\begin{eqnarray}\label{lltau}
 \tau=x (\s\wedge \s_{xx}+\s\wedge J\s)+\s\wedge \s_x,
\end{eqnarray}
which is the same as the one given in \cite{Fuch84}. It generates the
hierarchy of symmetries of every order starting from $a_0=\s_x$.  

These symmetries can also be produced by recursion operators.
For the Landau-Lifshitz equation, there exists two weakly nonlocal \cite{MaN01}
recursion operators of orders $2$ and $3$  and they are related by an elliptic
curve equation \cite{DS08}.  We need to use both of them to generate all these
symmetries obtained from the
master symmetry (\ref{lltau}).

We now compute the time dependent symmetries for equation (\ref{lleq}). First
we have
\begin{eqnarray}\label{conw}
\ad_{\bar K}w=-\frac{1}{10} \ad_f^2 {\bar K},
\end{eqnarray}
where
\begin{eqnarray*}
{\bar K}= \left( \begin{array}{c}K\\0
\end{array}\right), \quad w= \left( \begin{array}{c}
\frac{1}{3} x^3 \s_x+x^2 \s\\\frac{1}{3}x^3 J_x+2 x^2J \end{array}\right),
\quad f=-\left( \begin{array}{c} x^2  \s_x +2 x \s\\x^2 J_x+4 x J
\end{array}\right).
\end{eqnarray*}
Notice that the second component of $ \ad_{\bar K} w $ is zero. So are all
$\ad_{\bar K} \ad_f^n w $. We denote the first components of them by $w_n $ and
$\ad_K^{n+3} w_n=0$. According to Theorem \ref{thm3}, the time dependent
symmetries for equation (\ref{lleq}) are 
\begin{eqnarray*}
\frac{\partial^j }{\partial t^j} \exp(-t \ad_K) w_n,
\quad \mbox{where}\ n=0, 1, 2, \cdots, \mbox{and} \ j=0, 1, 2, \cdots, n+1.
\end{eqnarray*}

\subsubsection{A Benjamin-Ono type equation}\label{sec43}
In \cite{mn2}, Mikhailov and Novikov classified the Benjamin-Ono-type equations
with higher symmetries using
the perturbative symmetry approach in symbolic representation. Here we construct a master symmetry for the following equation in their list:
\begin{equation}\label{BOc}
u_t=K=H(u_{xx})+D_x \left(\frac{1}{2}c_1 u^2+\frac{1}{2} c_2 H(u^2)-c_2 u H(u)
\right),
\end{equation}
where $c_1$ and $c_2$ are constant and $H$ denotes the Hilbert transform
\begin{equation*}
\label{hilbert} H(f)=\frac{1}{{\rm i}\pi}\int_{-\infty}^{\infty}\frac{f(y)}{y-x}dy, \qquad {\rm i}^2=-1.
\end{equation*}
Notice that when $c_1=2$ and $c_2=0$, it reduces to the well-known Benjamin-Ono
equation:
\begin{equation}\label{BO}
u_t=K=H(u_{xx})+2uu_x.
\end{equation}
The higher order symmetries and conservation laws of the Benjamin-Ono
type equations contain nested Hilbert transform and thus  an appropriate
extension of the differential ring is required. We refer \cite{mn2} for the details of the extension.
It is very similar to the one proposed by Mikhailov and Yamilov in \cite{mr1643816}
for $(2+1)$-dimensional equations, which we will describe in section
\ref{sec52}.

We now apply Theorem \ref{thm2} to construct a master symmetry of (\ref{BOc}). First we notice that
it is homogeneous and we have
\begin{eqnarray*}
 [x u_x+u, K]=2 K.
\end{eqnarray*}
In \cite{mn2}, it was proved that the symbolic representations of its symmetries start with linear term in $u$.
According to Remark \ref{rem1}, condition (ii) in Theorem \ref{thm2} is satisfied. Let
\begin{eqnarray}\label{tauBO}
 \tau=\frac{1}{4}[x^2 u_x+2 x u, K]=x K+\frac{3}{2} H(u_x)+\frac{1}{2}c_1
u^2+\frac{1}{2} c_2 H(u^2)-c_2 u H(u).
\end{eqnarray}
We have
\begin{eqnarray*}
&& a_2=[\tau, K]=2 u_{xxx}+\frac{c_1}{2}D_x\left(3  H(u u_x)+3  u H(u_x)+c_1 u^3 \right)\\
&&\qquad +\frac{c_2}{2}D_x\left(-3  H(u_x H(u))-3  H(u) \ H(u_x)-3 c_1  u^2 H(u)+c_1  H(u^3)\right)\\
&&\qquad +\frac{c_2^2}{2}D_x\left( u^3-3  H(u^2 H(u))+3  u (H(u))^2\right)
\end{eqnarray*}
and $[a_2, K]=0$. Here we take into account the relation $H^2 = 1$ and the Hilbert-Leibnitz rule
$$H(fg) = fH(g) + gH(f)- H(H(f)H(g)).$$
Thus $\tau$ given by (\ref{tauBO}) is a master symmetry. 
When $c_1=2$ and $c_2=0$, we get the master symmetry of (\ref{BO}) in \cite{mr83f:58039}.

As we discussed in Section \ref{sec3}, we can use the master symmetry to construct conserved densities. Equation (\ref{BOc}) is Hamiltonian
with a Hamiltonian operator ${\cal H}=D_x$. Indeed,
\begin{eqnarray*}
 u_t=D_x \frac{\delta \rho_1}{\delta u}=D_x \left(H(u_x)+\frac{1}{2}c_1
u^2+\frac{1}{2} c_2 H(u^2)-c_2 u H(u) \right), 
\end{eqnarray*}
where
\begin{eqnarray*}
 \rho_1=\frac{1}{2} \left(u
H(u_x)+\frac{1}{3} c_1 u^3+c_2 u H(u^2)\right). 
\end{eqnarray*}
Moreover, we have
$$
 \tau=D_x\left(x\frac{\delta \rho_1}{\delta u}+\frac{1}{2} H(u)\right)=D_x
\frac{\delta (x \rho_1)}{\delta u} .
$$
Let $\rho_0=\frac{1}{2}u^2$, which corresponds to $a_0=u_x=D_x \frac{\delta
\rho_0}{\delta u}$. It follows from Proposition \ref{pro3} that the symmetries
$a_n=[\tau,\ a_{n-1}]$ are Hamiltonian vector fields and 
$a_n=D_x \frac{\delta \rho_n}{\delta u}$, where $\rho_n\equiv
{\rho_{n-1}}_*(\tau)$.  We can check that
$
\rho_1\equiv{\rho_0}_*(\tau)$ and find
\begin{eqnarray*}
&&\rho_2\equiv{\rho_1}_*(\tau)\equiv \frac{\delta ( \rho_1)}{\delta
u} D_x\left(\frac{1}{2} H(u)+x\frac{\delta ( \rho_1)}{\delta u}
\right)\equiv \frac{1}{2}\frac{\delta ( \rho_1)}{\delta
u}\left(H(u_x)+\frac{\delta ( \rho_1)}{\delta u} \right)\\&&\quad
\equiv-u_x^2+\frac{c_1}{4} \left(3 u^2 H(u_x)+\frac{c_1}{2} u^4\right)
+\frac{c_2}{2}\left(-3 u H(u) \ H(u_x)- c_1  u^3 H(u)\right)\\&&\qquad 
+\frac{c_2^2}{4}\left(3  u^2 (H(u))^2+\frac{u^4}{2}\right) .
\end{eqnarray*}
When $c_1=2$ and $c_2=0$, we get the known conservation laws for the Benjamin-Ono
equation (\ref{BO}) \cite{mr83f:58039}.

\subsection{$(2+1)$-dimensional partial differential equations}\label{sec52}
The method of constructing master symmetries proposed in Section \ref{sec4} is
valid for $(2+1)$-dimensional partial differential equations. A typical example
is the Kadomtsev-Petviashvili (KP) equation
\begin{eqnarray}\label{kp}
u_t=u_{xxx}+6 u u_x +3 D_x^{-1} D_y u_{y},
\end{eqnarray}
where the dependent variable $u$ is a smooth
function of independent variables $x, y$ and $t$, and  $D_x^{-1}$ is the formal
inverse of the total $x$-derivative.
The main obstacle to directly extend the theories of
$(1+1)$-dimensional nonlinear evolutionary equations 
to $(2+1)$-dimensional equations is the nonlocality. The concept of conservation laws 
and the symmetry approach for testing integrability \cite{mr93b:58070} are mainly based on the locality.
In 1998,  Mikhailov and Yamilov addressed the nonlocality
problem in \cite{mr1643816}.
They noticed that the operators $D_x^{-1}$ and $D_y^{-1}$ never appear alone
but always in pairs like $D_x^{-1}D_y$ and $D_y^{-1}D_x$ for
all known integrable equations and their hierarchies of symmetries
and conservation laws. 
Based on this observation, they introduced the concept of quasi-local functions,
which is a natural generalization of local functions. 
Using the symbolic representation, it was proved this 
observation is true for integrable equations obtained from certain scalar Lax operators \cite{wang21}.

We denote the derivatives of dependent variable $u$ with respect to its
independent variables $x$ and $y$ by $u_{ix,jy}=\partial_x^i
\partial_y^j u$. For smaller $i$ and $j$, we sometimes write the indexes out
explicitly,
that is, we write $u_{xxy}$ and $u$ instead of $u_{2x,1y}$ and $u_{0x,0y}$. All smooth functions depending on $x, y, t, u$ and derivatives of $u$
form a differential ring $\A$ with total $x$-derivation and $y$-derivation
$$
D_x=\sum_{i=0}^{+\infty}\sum_{j=0}^{+\infty}u_{(i+1)x,jy}\frac{\partial}
{\partial u_{ix,jy}}
\quad \mbox{and}\quad
D_y=\sum_{i=0}^{+\infty}\sum_{j=0}^{+\infty}u_{ix,(j+1)y}\frac{\partial}
{\partial u_{ix,jy}}.
$$
Let us denote
\begin{eqnarray}\label{Theta}
\Theta=D_x^{-1}D_y, \qquad \Theta^{-1}=D_y^{-1}D_x.
\end{eqnarray}
To define  quasi-local functions $\A(\Theta)$, we consider a sequence of extensions of $\A$ as follows:

Let $\Theta \A=\{\Theta f:\
f \in \A\}$ and $\Theta^{-1} \A=
\{\Theta^{-1} f:\ f \in \A\}$. We define
$\A_0(\Theta )=\A$
and $\A_k (\Theta)$ is the ring closure of the union
$$\A_{k-1}(\Theta )\cup \Theta \A_{k-1}(\Theta )
\cup \Theta^{-1} \A_{k-1}(\Theta ).$$
Here the index $k$ indicates the maximal depth
of nesting the operator $\Theta^{\pm1}$ in the expression.
Clearly, we have $\A_{k-1}(\Theta )\subset \A_k(\Theta )$.
We now define $\A(\Theta)=\lim_{k\to \infty}\A_k (\Theta)$.
However, for a given $f\in \A(\Theta)$, there exists $k$
such that $f\in \A_k(\Theta)$. Note that $\A(\Theta)$ is not invariant under
transformations of variables.

Since the dependent variable depends on two spatial variables $x$ and $y$, we consider the scaling
symmetry as $x u_x +\beta y u_y +\alpha u$, where $\alpha$ and $\beta$ are constant. As a convention, we take the weight of $x$-derivative to be $1$.
Around $h$ we
can build two $\liesl(2)$ representations by choosing either $e=u_x$ or $e=u_y$. We present them in following two lemmas,
which is similar to Lemma \ref{lem1} and Lemma \ref{lem2} for the
$(1+1)$-dimensional case.
\begin{Lem}\label{lem5}
Given a scaling $h=2 (x u_x +\beta y u_y+\alpha u)$, where $\alpha$ and $\beta$ are constant, 
the elements $e=u_x$ and $f=-(x^2 u_x+2 \beta x  y u_y+ 2 \alpha x u)$ form an $\liesl(2, \C)$ with $h$.
Moreover, there exists  $w=\frac{1}{3} x^3 u_x+ \beta x^2  y u_y+ \alpha x^2 u$. We have $\ad_e w=f$ and 
\begin{eqnarray}\label{adfn1}
\ad_f^n w=\frac{n!}{3} \left(x^{n+3} u_x +(n+3) x^{n+2} (\beta y u_y+\alpha u)\right), \quad n=0, 1, 2,\cdots 
\end{eqnarray}
\end{Lem}
The proof of this lemma is the same as we did for Lemma \ref{lem1} and Lemma \ref{lem2}. We won't repeat it again.

Notice that the role of $x$ and $y$ is equal. We can alternatively build up the $\liesl(2, \C)$ by taking $e=u_y$ instead.
\begin{Lem}\label{lem6} Assume that $\beta\neq 0$.
Given a scaling $h=\frac{2}{\beta} (x u_x +\beta y u_y+\alpha u)$, where $\alpha$ and $\beta$ are constant, when 
the elements $e=u_y$ and $f=-\frac{1}{\beta}(2 x y u_x+\beta  y^2 u_y+ 2 \alpha y u)$ form an $\liesl(2, \C)$ with $h$.
Moreover, there exists  $w=\frac{1}{3} y^3 u_y+ \frac{1}{\beta} y^2 ( x u_x+ \alpha u)$. We have $\ad_e w=f$ and 
\begin{eqnarray}\label{adfn2}
\ad_f^n w=\frac{n!}{3} \left(y^{n+3} u_y +\frac{n+3}{\beta} y^{n+2} ( x u_x+\alpha u)\right), \quad n=0, 1, 2,\cdots 
\end{eqnarray}
\end{Lem}

The nonlocality can cause the invalidity of the Jacobi identity for the
characteristics
of nonlocal vector fields. This was first noticed in \cite{mr1874291} when the authors systematically investigate the symmetry
properties of the KP equation. To solve this problem,
the notion of a ghost characteristic was introduced in \cite{mr1900193,olvgho}.
The ghost characteristics are the expressions in the kernels of $D_x$ and/or of $D_y$.
One of advantages of $\liesl(2)$ representations presented in Lemma \ref{lem5} and \ref{lem6} is that  there are no such terms and thus the ghost 
problems will not appear.

\subsubsection{The Kadomtsev--Petviashvili equation}\label{seckp}
We use our notation and rewrite the KP equation (\ref{kp}) as follows:
\begin{eqnarray}\label{kpth}
u_t=K=u_{xxx}+6 u u_x +3 \exnl u_{y}. 
\end{eqnarray}
If ignoring the dependence of $y$, the equation reduces to the well-known
Korteweg-de Vries (KdV) equation. As a natural generalisation of
the KdV equation, its symmetry structure has been well studied in the literature. For example, its infinitely many $t$ dependent symmetries
are given in \cite{mr85m:58187, mr88g:58195}. Based on this,  we presented its three different $\liesl(2)$ representations in \cite{mr1874291},
similar to the ones we discussed for the Burgers equation in section \ref{sec41}. Here we revisit this equation
and construct its master symmetries using our approach using both $\liesl(2)$ representations described in Lemma \ref{lem5} and \ref{lem6}.

The KP equation is homogeneous for scaling $h=2 y u_y+x u_x+2 u$ and
\begin{eqnarray*}
 [h, K]=3 K.
\end{eqnarray*}
It follows from Lemma \ref{lem5} and \ref{lem6} that there are two $\liesl(2, \C)$ representations around $h$. 
When we use the $\liesl(2, \C)$ in Lemma \ref{lem6}. 
According to Theorem \ref{thm2}, we define  
\begin{eqnarray*}
\tau=\frac{1}{6} [x y u_x+  y^2 u_y+ 2 y u,\  K]
=\frac{1}{2} y (u_{xxx}+6 u u_x +3 \exnl u_{y})+x u_y +2 \exnl u,
\end{eqnarray*}
which leads to the master symmetry for the KP equation in \cite{oef82}.
Indeed, the next symmetry, denoted by $P$, is
\begin{eqnarray}\label{kpsym2}
 P=[\tau, K]=6 (u_{xxy}+\Theta^2 u_y+4 u u_y +2 u_x \Theta u) .
\end{eqnarray}

If we use the $\liesl(2, \C)$ in Lemma \ref{lem5}, we have $\ad_K^3 f=\ad_K^3 (x^2 u_x+4 x  y u_y+ 4 x u)\neq 0 $. This implies that
$[f, K]$
is not a master symmetry. As we mentioned in Remark \ref{rem1}, $f'=f+f_0$ also forms $\liesl(2)$ with elements $e$ and $h$ when $[e,f_0]=0$ and $[h, f_0]=-2 f_0$.
Using the condition $\ad_K^3 f'=0$, we are able to determine that
$$f_0=y^2 (u_{xxx}+6 u u_x +3 \exnl u_{y}) +8 y \exnl u +2 D_x^{-1} u$$
and 
\begin{eqnarray*}
f'=f+ f_0=2 y \tau +x^2 u_x +2 xy u_y +4 x u +4 y \exnl u +2 D_x^{-1} u  .
\end{eqnarray*}
This leads to
\begin{eqnarray}\label{kpmas2}
\tau'=\frac{1}{6} [f',\  K]=\frac{2}{3} y P+x K+4 u_{xx}+8 u^2+2 u_x D_x^{-1}u +6 \exnl^2 u ,
\end{eqnarray}
where $P$ is defined by (\ref{kpsym2}). This lies in the non-isospectral KP hierarchy given in \cite{zhang13} and it is a master symmetry.
It is easy to see that
\begin{eqnarray*}
 [\tau,\ u_x]=u_y \qquad \mbox{and} \qquad [\tau',\ u_x]=K. 
\end{eqnarray*}
Starting from $u_x$,
the master symmetry $\tau'$ generates the hierarchy of symmetries with the reduction to the symmetry hierarchy of the KdV equation
if $u$ is independent of $y$. 

In the following Diagram \ref{dia3}, we demonstrate the relations between these two $\liesl(2)$ and the corresponding
master symmetries. 
\begin{figure}[hz]
\captionsetup{name=Diagram}
\centering
\begin{diagram}
& u_x       &        &           &          &            &             &        &           &\\
& \dTo^{\tau}   & \rdTo(2,4)~{f'} &           &          &            &             &        &           &\\
& u_y        &         &          &          &            &             &        &           &\\
& \dTo^{\tau}   & \rdTo_f &           &          &            &             &        &           &\\
 & K       &         &  h        &           &           &             &         &           &\\
&\dTo^{\tau} &\rdTo(2,4)~{f'}  \rdTo^f &   \dTo_{\tau}        &\rdTo(2,4)~{f'} \rdTo^f  &            &             &         &           &\\
& P        &         & \tau      &          & f          &             &         &            &\\
&\dTo^{\tau} & \rdTo_f &          & \rdTo_f  &            &             &         &           &\\
&         &         & \tau'      &          & f'         &             &         &            &
\end{diagram}
\caption{Master symmetries for KP equation and two $\liesl(2, \C)$}
\label{dia3}
\end{figure}
\subsubsection{The noncommutative KP equation}
Consider the noncommutative KP equation \cite{mr93c:58085,wang21}
\begin{eqnarray}\label{nkp}
u_{t}=K=u_{xxx} +3 u u_x +3 u_x u +3 \Theta u_y -3 C_u \Theta u ,
\end{eqnarray}
where the dependent variable $u$ takes its value in an associative algebra and
$C_u$ denotes the commutator in the associative algebra, that is, $C_u
\Theta u=u \Theta u-(\Theta u) u$, which is zero if $u$ takes its value in a
commutative algebra.  When $u$ is independent of $y$, it leads to the
noncommutative KdV equation ( see \cite{mr99c:58077, mr1781148} for more
examples and noncommutative $(1+1)$-dimensional
integrable evolution equations)
\begin{eqnarray*}
 u_{t}=K=u_{xxx} +3 u u_x +3 u_x u .
\end{eqnarray*}
The noncommutative KP equation (\ref{nkp}) is homogeneous with respect to the
same scaling $h=2 y u_y+x u_x+2 u$ for the KP equation (\ref{kpth}).

Notice also that the elements in $\liesl(2)$ presented in Lemma \ref{lem5} and \ref{lem6} are linear in $u$ and its derivatives, which implies that
they are valid no matter the dependent variables are commutative or noncommutative. Therefore, we can apply Theorem \ref{thm2} to equation (\ref{nkp}).
We define
\begin{eqnarray}
&&\tau=\frac{1}{6} [x y u_x+  y^2 u_y+ 2 y u,\  K]\nonumber\\
&&\quad=\frac{1}{2} y (u_{xxx}+3 u u_x +3 u_x u +3 \Theta u_y -3 C_u \Theta u)+x u_y +2 \exnl u-\frac{1}{2} C_u D_x^{-1} u.\label{taunkp}
\end{eqnarray}
Using it, we compute
\begin{eqnarray*}
&& [\tau, K]=6 (u_{xxy}+\Theta^2 u_y+2 u u_y +2 u_y u+u_x \Theta u
+(\Theta u) u_x)\\
&&\qquad\quad +6(C_u D_x^{-1} C_u \Theta u
-\Theta C_u \Theta u-C_u \Theta^2 u),
\end{eqnarray*}
which is a symmetry flow as presented in \cite{wang21} implying that $[[\tau, K], K]=0$. Condition (ii) in Theorem \ref{thm2} is trivially satisfied
due to Remark \ref{rem3}. Thus $\tau$ defined by (\ref{taunkp}) is a master symmetry
for equation (\ref{nkp}).

\subsubsection{A new integrable Davey-Stewartson type equation}
Recently, Huard and Novikov carried out the classification of integrable Davey-Stewartson type equations \cite{HNov13}. They found a few new equations.
Here we use our approach to construct a master symmetries for the one (equation (3.6) in \cite{HNov13}), whose linear terms have constant coefficients. 
We further give its Hamiltonian operator and 
compute its conserved densities using the master symmetry.

Consider the system of the following form
\begin{eqnarray}\label{ds}
\left\{ \begin{array}{l} u_t=u_x \exnl^{-1} u-\frac{1}{4} c^2 v v_y +v_x +\epsilon (u_{xx}+c v_{xy})
\\ v_t=D_x( v \exnl^{-1}u)+c v v_x-\epsilon v_{xx} \end{array} \right.
\end{eqnarray}
where $c$ and $\epsilon$ are constant. It is homogeneous since
we have
\begin{eqnarray*}
 \big{[}\left(\begin{array}{c}x u_x\\ x v_x+v \end{array}\right), \left(\begin{array}{c} u_t\\ v_t \end{array}\right)\big{]}=2 \left(\begin{array}{c} u_t\\ v_t \end{array}\right).
\end{eqnarray*}
To shorten the expression, we use
notations $p_y=u_x$ and $q_y=v_x$. According to Theorem \ref{thm2}, its master symmetry is
\begin{eqnarray*}
\tau= \frac{1}{4}\big{[}\left(\begin{array}{c}x^2 u_x\\ x^2 v_x+2 x v \end{array}\right), \left(\begin{array}{c} u_t\\ v_t \end{array}\right)\big{]}
 =\left(\begin{array}{c}  x u_t+\frac{\epsilon}{2} (u_x+c v_y)+\frac{v}{2} \\  x v_t+ v p +\frac{c}{2} v^2-\frac{3 }{2} \epsilon v_x \end{array}\right) .
\end{eqnarray*}
Notice that the anti-symmetric constant operator
\begin{eqnarray*}
\cal H= \left(\begin{array}{cc} -c D_y^2 D_x^{-1} & D_y \\ D_y & 0 \end{array}\right)
= \left(\begin{array}{cc} -c D_y \Theta & D_y \\ D_y & 0 \end{array}\right)
\end{eqnarray*}
is Hamiltonian, and it is a Hamiltonian operator for equation (\ref{ds}) since we can write it as
\begin{eqnarray}
&& \left(\begin{array}{c}u_t\\ v_t \end{array}\right)={\cal H} \left(\begin{array}{c}\exnl^{-1} (v p)+\frac{c}{2} \exnl^{-1}v^2- \epsilon \exnl^{-1} v_x \\ 
 \frac{1}{2}p^2+c v p+\frac{3}{8} c^2 v^2+ q+\epsilon p_x 
\end{array}\right)\nonumber\\
 &&\qquad={\cal H} \delta_{(u,v)} \left(\frac{1}{2} (v p^2+c v^2 p+v q)+\epsilon
v p_x +\frac{1}{8} c^2 v^3\right).\label{eqdsh}
\end{eqnarray}
We now use the master symmetry to construct the hierarchy of conserved densities of (\ref{ds}).
It is clear $u$ and $v$ are two conserved densities since we can write
\begin{eqnarray*}
&&u_t=D_y \left(\frac{1}{2} p^2-\frac{1}{8}c^2 v^2+q +\epsilon (p_x+c v_x) \right);\\
&&v_t=D_x\left( v p+\frac{c}{2} v^2-\epsilon v_{x}\right).
\end{eqnarray*}
However, we can't take the above two densities as starting points since
we have $\int u_*(\tau)=v$ and $\int v_*(\tau)=0$. Besides, they both generate
zero Hamiltonian vector field.  Notice that
\begin{eqnarray*}
&& \left(\begin{array}{c}u_x\\ v_x \end{array}\right)={\cal H} \left(\begin{array}{c}q \\ 
 p+c v \end{array}\right)={\cal H} \delta_{(u,v)} \left(v p +\frac{c}{2} v^2 \right), 
\end{eqnarray*}
and
\begin{eqnarray*}
&& \tau={\cal H} \delta_{(u,v)} \left(\frac{x}{2} (v p^2+c v^2 p+v
q)+\epsilon x v p_x +\frac{1}{8} c^2 x v^3+\frac{\epsilon}{2}u q\right). 
\end{eqnarray*}
Let $\rho_0=v p +\frac{c}{2} v^2 $. We apply Proposition \ref{pro3} to
construct Hamiltonians corresponding to symmetries, that is, $a_n=[\tau,
a_{n-1}]={\cal H} \delta_{(u,v)} \rho_n$ and $\rho_n=\int {\rho_{n-1}}_*(\tau)$.
In particular, we have 
%\begin{eqnarray*}
% &&D_t\rho_0=v_t p+v p_t+c v v_t\\
% &&\quad =D_x\left( ( v p+\frac{c}{2} v^2-\epsilon v_{x}) p +v
%\left(\frac{1}{2} p^2-\frac{1}{8}c^2 v^2+q+\epsilon (p_x+c v_x)\right)\right)\\
%&&\quad+D_x \left( c v^2 p+\frac{3 c^2}{8} v^3-\epsilon c v v_{x}-\frac{c}{2}
%v^2 p-\frac{v}{2} p^2\right)-\frac{1}{2} D_y  q^2.
%\end{eqnarray*}

\begin{eqnarray*}
&&\rho_1=\int {\rho_0}_* (\tau)=
\int  \left( (x v_t+ v p +\frac{c}{2} v^2-\frac{3 }{2} \epsilon v_x) p
+v \exnl^{-1}(x u_t+\frac{\epsilon}{2} (u_x+c v_y)+\frac{v}{2}) \right)\\
&&\quad +\int c v (x v_t+ v p +\frac{c}{2} v^2-\frac{3 }{2} \epsilon v_x)\\
&&\quad\equiv \frac{1}{2} \left(v p^2+c v^2 p+v q\right)+\epsilon v p_x
+\frac{1}{8} c^2 v^3,
\end{eqnarray*}
which is the Hamiltonian for this equation (see (\ref{eqdsh})),
and
\begin{eqnarray*}
&& \rho_2=\int {\rho_1}_* (\tau)\equiv \frac{1}{2}\left( v p^3+3 v p q-3 \epsilon v_x p^2
-4 \epsilon^2 p_x v_x\right)+\frac{3}{4} c v^2 \left(q+\epsilon p_x+p^2\right)
\\&&\qquad
- c \epsilon^2 v_x^2+\frac{3}{8} c^2 v^3 p +\frac{1}{16} c^3 v^4,
\end{eqnarray*}
which corresponds to the next symmetry
\begin{eqnarray*}
&&a_2={\cal H} \delta_{(u,v)} \rho_2\\
&&= \left(\begin{array}{l} \begin{array}{r}D_y \left(\frac{1}{2} p^3+\frac{3}{2}
p q +3 \epsilon p p_x+2 \epsilon^2 p_{xx}+\frac{3}{2} \epsilon c  (p_x v +2 p
v_x)-\frac{1}{8} c^2 v^2(3 p+c v)\right)+\\+D_y\left(
 \frac{3}{2} \exnl^{-1}(p v)+\frac{3}{4} c \exnl^{-1} v^2 \right)\end{array}\\ 
\frac{3}{2} D_x \left(p^2 v + v q-2 \epsilon p v_x +\frac{4}{3} \epsilon^2 v_{xx}+ c p v^2- \epsilon c v v_x+\frac{1}{4} c^2 v^3 \right)        
 \end{array} \right).
\end{eqnarray*}

\subsection{$(2+1)$-dimensional lattice-field  equations}
In this section, we apply the proposed approach in section \ref{sec4} to $(2+1)$-dimensional lattice-field  equations.
Consider the differential-difference KP equation of form \cite{djm82}
\begin{eqnarray}\label{ddkp}
u_{t}=u_{yy} +2 u u_y+2 (\cS-1)^{-1} u_{yy},
\end{eqnarray}
where the dependent variable $u=u(n,y,t)$ is a function of continuous variables $y$, $t$ and 
discrete variable $n\in \Z$, and it is smooth with respect to the continuous variables.
Here $\cS$ is the shift operator mapping $u(n,y,t)$ to $u(n+1,y,t)$. The operator $\cS-1$ is a discrete analogue
of derivative. Notice that $D_y$ and $(\cS-1)^{-1}$ appear in pair. In the same way as in the case for $(2+1)$-dimensional partial differential equations described
in Section \ref{seckp}, we also introduce the concept of quasi-local functions.

We denote  $\cS^i \partial_y^j u$ by $u_{i,jy}$, where $i\in \Z$ and $j\in \N$
is the order of derivative of the dependent variable $u$ with respect to its
independent variable $y$. When $i=0$ or $j=0$, we simply write as $u_i$ or $u_{jy}$. All functions depending on $n, y, t, u$ and $u_{i,jy}$ and
being smooth with respect to its variables except $n$ form a differential ring $\A$ with total $y$-derivation $D_y$.
Let us denote
\begin{eqnarray}\label{Thetas}
\Theta=(\cS-1)^{-1} D_y, \qquad \Theta^{-1}=D_y^{-1} (\cS-1).
\end{eqnarray}
By considering a sequence of extensions of $\A$ as in Section \ref{seckp}, we can define the quasi-local functions $\A(\Theta)$.

Since $y$ is the only continuous spatial variable for dependent variable, we consider the scaling
symmetry as $ y u_y +\alpha u$, where $\alpha$ is constant. 
Thus, around $h$ we get the same $\liesl(2)$ representations (only changing $x$ to $y$) as presented 
in Lemma \ref{lem1} and Lemma \ref{lem2} for the $(1+1)$-dimensional case.
We are then ready for applying Theorem \ref{thm2} to these type of evolutionary homogeneous equations.

\subsubsection{The differential-difference KP equation}
Using our notation, we rewrite equation (\ref{ddkp}) as
\begin{eqnarray}\label{ddkpth}
u_{t}=K=u_{yy} +2 u u_y+2 \Theta u_{y} .
\end{eqnarray}
Its master symmetry has recently been studied in \cite{zhang13, farbod14}. In the latter paper \cite{farbod14}, the author also studied its time dependent 
symmetries using the $\liesl_h(2)$ (cf. Diagram \ref{dia2}) representation. Here
we demonstrate our approach on how to construct the master symmetry.
Notice that equation (\ref{ddkpth}) is homogeneous with respect to $h=y
u_y+u$ and
\begin{eqnarray*}
[y u_y +u, K]=2 K.
\end{eqnarray*}
Following Remark \ref{rem1} we take $f=-(y^2 u_y+2 y u +f_0)$ forming
$\liesl(2)$ with $u_y$ and $h$, where $f_0$ satisfies $[e, f_0]=0$ and $[h,
f_0]=-2 f_0$. Assume $f_0$ is a function of independent variable $n$,
i.e., $f_0=g(n)$ and 
define
\begin{eqnarray*}
\tau=\frac{1}{4}[y^2 u_y+2 y u+g(n), K]=y(u_{yy}+2 u u_y+2 \Theta u_y)
+\frac{1}{2} g(n)
u_y+\frac{3}{2} u_y+u^2 +3 \Theta u.
\end{eqnarray*}
Using the condition $[[\tau, K], K]=0$, we can determine $g(n)=2n$. The
resulting $\tau$ is equivalent to the master symmetry presented in
\cite{zhang13, farbod14} in the sense that it generates the same space of
symmetries for equation 
(\ref{ddkpth}). 

\subsubsection{The (2+1)-dimensional Volterra lattice}\label{sec531}
Consider the following differential-difference equation
\begin{eqnarray}
 &&u_t=K=(\cS+1)(\cS-1)^{-1}u_{yy}+u_y (\cS+1)(\cS-1)^{-1}u_y+\exp(2u_{1})-\exp(2 u_{-1})\nonumber\\
 &&\ \quad \ \qquad=(\cS+1)\exnl u_y+u_y (\cS+1)\exnl u +\exp(2u_{1})-\exp(2 u_{-1}) .\label{vol1} 
\end{eqnarray}
This system of periodic form is called two dimensional generalisation
of the Volterra lattice \cite{mik79}. Its Lax representation is invariant under dihedral reduction groups \cite{mik79,mik80,mik81},
which is generated by both inner and outer automorphisms
of $\liesl(n,\C)$, and it can be viewed as a discretisation of
the Kadomtsev-Petviashvili equation \cite{LM04}. For fixed period $n$, it is a bi-Hamiltonian system.
When $n=3$, its recursion operator and bi-Hamiltonian structure are explicitly constructed from its Lax representation in 
\cite{wang09}. Its Darboux transformation for arbitrary period $n$ has recently
been constructed in \cite{MPW}.

If dependent variable $u$ is independent of $y$, equation (\ref{vol1}) reduces
into 
$$u_t=\exp(2u_{1})-\exp(2 u_{-1}),$$
which is the well-known Volterra chain $v_{t'}=v (v_1-v_{-1})$ under the point
transformation $v=\exp(2u)$ and $t'=2t$.

In this section, we construct a master symmetry for (\ref{vol1}) using our approach in Section \ref{sec4} and further construct its conserved
densities via master symmetries.

Notice that equation (\ref{vol1}) is homogeneous with respect to scaling $y
u_y+1$. Indeed, by direct computation, we have $[y u_y+1, K]=2 K$ for equation
(\ref{vol1}).  The following three elements form an $\liesl(2)$:
\begin{eqnarray*}
 e=u_y,\quad h=2(y u_y+1)\quad  \mbox{and}\quad f=-(y^2 u_y+2y).
\end{eqnarray*}
It follows from Theorem \ref{thm2} that
\begin{eqnarray}\label{msvol1}
 \tau=\frac{1}{4}[y^2 u_y+2y, K]=y K+n u_y+ 2 \exnl u 
\end{eqnarray}
is a master symmetry for equation (\ref{vol1}) if it satisfies the conditions. First of all, Condition (ii) of Theorem
\ref{thm2} is trivially satisfied according to  Remark \ref{rem3}. We now check Condition (i) $[[\tau, K], K]=0$.
By direct calculation, we have
\begin{eqnarray*}
&&a_2=[\tau, K]= (\cS^2+\cS+1) \exnl^2 u_y+3 u_y \cS^2 \exnl^2 u+3 (\exnl u) (\cS+1)\exnl u_y+u_y^3+3 u_y (\exnl u) \cS (\exnl u)\\
&&\qquad+\frac{3}{2} D_y( \exp(2u_1) +\exp(2u)+\exp(2u_{-1}))+3 \exp(2u_1) u_y\\
&&\qquad+3 \left(\exp(2u_1)-\exp(2 u_{-1})\right) \exnl u,
\end{eqnarray*}
where we used the identity
$$ (\cS-1)^{-1} (n u_y)=(n-1) (\cS-1)^{-1} u_y-(\cS-1)^{-2} u_y
$$
and further $[a_2, K]=0$. Here we skip the long formulas to check $a_2$ is a symmetry of (\ref{vol1}),
which can be carried out by organising the terms according to polynomial terms, exponential terms and mixed terms.
We can also compare $a_2$ to the symmetry flows obtained via its Lax representation, which is given in \cite{mik79, LM04, FNR13}.

We now look at conserved densities for equation (\ref{vol1}). Equation (\ref{vol1}) can be written as
\begin{eqnarray*}\label{vol}
 &&u_t=K=(\cS-1) \left( (\cS+1) \exnl^2 u +(\exnl u)^2+ \exp(2u)+\exp(2 u_{-1}) \right)\nonumber,
\end{eqnarray*}
which implies that $u$ is a conserved density. Since $\int u_*(\tau)=0$, we can not use it as a start point to generate other conserved densities.

Let $\rho_0=\frac{1}{2} (\exnl u)^2+\exp(2u)$, which is a conserved density for (\ref{vol1}). Indeed, we have
\begin{eqnarray*}
&& \frac{\partial \rho_0}{\partial t}=(\exnl u) \exnl u_t +2 \exp(2 u) u_t\\
&&\quad =D_y\left( (\exnl u) \left( (\cS+1) \exnl^2 u +\frac{2}{3} (\exnl u)^2+ 3 \exp(2u)+\exp(2 u_{-1})  \right)+\exp(2u) u_y \right) \\
&&\qquad+(\cS-1)\left( \exp(2u_{-1}) \exnl u_y + 2 \exp(2u+2u_{-1})-(\exnl^2 u )^2\right).
\end{eqnarray*}
For fixed period $n=3$, it reduces to the conserved density appeared in \cite{wang09}.
Starting from it, using the master symmetry given by (\ref{msvol1}) we get
the next conserved density:
\begin{eqnarray*}
&&\rho_1=\int {\rho_0}_*(\tau)=\int \tau \cdot \frac{\delta \rho_0}{\delta u}\equiv \frac{1}{3}(\exnl u)^3+ (\exnl u) (\exnl^2 u) +2 (\exnl u ) \exp(2u).
\end{eqnarray*}

\subsubsection{Another (2+1)-dimensional generalised Volterra
Chain}\label{sec532}
In this section, we construct a master symmetry for the following equation
\begin{eqnarray}
&& u_t=(\cS+1)(\cS-1)^{-1}u_{yy}+u_y (\cS+1)(\cS-1)^{-1}u_y+2\exp(u+u_{1})-2\exp(u+ u_{-1})\nonumber\\
 &&\ \quad=K=(\cS+1)\exnl u_y+u_y (\cS+1)\exnl u +2\exp(u+u_{1})-2\exp(u+ u_{-1}).\label{vol2}
\end{eqnarray}
This equation is appeared in \cite{FNR13} when the authors classified a family of equations with the non-locality of 
intermediate long wave type. In fact, both equations (\ref{vol1}) and (\ref{vol2}) are listed in 
the list of integrable equations of this class, and these two equations have coinciding dispersionless limits.

We can also consider the equation of periodic form. The corresponding equations of (\ref{vol1}) and (\ref{vol2}) for period $n=3$ 
appeared in the classification of integrable systems of nonlinear
Schr{\"o}dinger type \cite{mr89g:58092}.

If dependent variable $u$ is independent of $y$, equation (\ref{vol2})
reduces into 
$$u_t=2\exp(u+u_{1})-2\exp(u+ u_{-1}),$$
which is also  the well-known Volterra chain $v_{t'}=v (v_1-v_{-1})$ under the
point
transformation $v=\exp(u+u_1)$ and $t'=2t$.

Notice that equation (\ref{vol2}) is also homogeneous and shares the same scaling as equation (\ref{vol1}). Thus we can use the same $\liesl(2)$ to construct
its master symmetry. Here we only write out the results since the arguments of applying Theorem \ref{thm2} are the same due to equations (\ref{vol1}) and
(\ref{vol2}) having the same linear terms. We only need to check Condition (i) $[[\tau, K], K]=0$.

We define
\begin{eqnarray}\label{msvol2}
 \tau=\frac{1}{4}[y^2 u_y+2y, K]=y K+n u_y+ 2 \exnl u ,
\end{eqnarray}
which is of the same form as (\ref{msvol1}),
and  we have
\begin{eqnarray*}
&&a_2=[\tau, K]= (\cS^2+\cS+1) \exnl^2 u_y+3 u_y \cS^2 \exnl^2 u+3 (\exnl u) (\cS+1)\exnl u_y+u_y^3+3 u_y (\exnl u) \cS (\exnl u)\\
&&\qquad +\exp( u +u_1) (3 u_{1,y} +6 u_y+6 \exnl u)+ \exp( u +u_{-1}) (3 u_{-1,y}-6 \exnl u) .
\end{eqnarray*}
To verify $[a_2, K]=0$, we can either compute directly or compare it to the symmetry flows obtained via its Lax representation \cite{FNR13}.

We now look at conserved densities for equation (\ref{vol2}) and notice that 
$$\rho_0=\frac{1}{2} (\exnl u)^2+\exp(u+u_{-1})$$
is a conserved density. Indeed, we have
\begin{eqnarray*}
&& \frac{\partial \rho_1}{\partial t}=(\exnl u) \exnl u_t +\exp( u+u_{-1}) (u_t+u_{-1,t})\\
&&\quad \equiv D_y\left( -\frac{1}{3} (\exnl u)^3+2 \exp(u+ u_{1}) \exnl u +\exp(u+u_{1}) u_y +\exp(u+u_{-1}) u_y  \right) \\
&&\qquad+(\cS-1)\left(\exp(u+u_{-1}) u_y  u_{-1,y} -2\exp(u+u_{-1}) u_y \exnl u-(\exnl^2 u )^2\right) .
\end{eqnarray*}
Starting from it, using the master symmetry given by (\ref{msvol2}) we get
its next conserved density:
\begin{eqnarray*}
&&\rho_1=\int {\rho_0}_*(\tau)=\int \tau \cdot \frac{\delta \rho_0}{\delta u}\equiv 
\frac{1}{3} (\exnl u)^3+ (\exnl u) (\exnl^2 u) +(\exp(u+u_1)+\exp(u+u_{-1})) \exnl u.
\end{eqnarray*}

\section{Discussion}
In this paper, we present a new structure called the $\cO$-scheme for
homogeneous evolutionary integrable equations. For an evolutionary vector field
$u_t=K[u]$ the Lie algebra of symmetries is the
kernel space of $ad_K$, the Lie algebra of master symmetries
is in the kernel of $ad_K^2$. Symmetries, master symmetries and higher order
nilpotent elements can neatly
be fit as elements of modules in the BGG category $\cO$.
This is based on the observation that there is a free
 $\liesl(2)$ representation and further an infinite dimensional module in the
BBG category related to such equations, whether integrable or not.
We prove that under technical conditions, it is enough for us to construct
master symmetries and organise time dependent symmetries
using the elements in this module. It also offers us an approach to construct
master symmetries.

The master symmetries for two $(2+1)$-dimensional Volterra Chains
are of the same form, see (\ref{msvol1}) and (\ref{msvol2}), since they have the
same scaling symmetry. Thus given a family of equations $u_t=K$ sharing the same
scaling symmetry, condition that $\ad_K^2 f$ commuting 
$K$ in Theorem \ref{thm2} can be used as a criterion for classifying integrable
equations.   It is worth to explore although
the calculation involved is tedious, in particular, for some
$(2+1)$-dimensional equations whose integrability condition is challenging to
formulate.

In this paper, we presented the $\cO$-schemes for both $(1+1)$- and
$(2+1)$-dimensional partial differential equations and $(2+1)$-dimensional
lattice-field equations with $(\cS-1)^{-1} D_y$  as nonlocal terms. Currently,
we are working on extension of such $\cO$-scheme to integrable
differential-difference and discrete systems.

The $\cO$-scheme in this paper is formulated in the case of
 $\liesl(2)$ modules. The extension of this
construction to algebras of higher rank is a promising direction of research
which would enable us to study much wider class of systems important in
applications including the Boussinesq equation, the
resonant wave interaction system, two dimensional Toda lattice and many others.

\section*{Acknowledgement}
The author would like to thank A.V. Mikhailov, J.A. Sanders and V.S. Novikov for
useful discussions,
and gratefully acknowledges financial support through EPSRC grant EP/I038659/1.
%\bibliography{kdv}

\begin{thebibliography}{99}

\bibitem{AKNS74}
Ablowitz M.J., Kaup D.J.,  Newell A.C., Segur H.:
\newblock Inverse scattering transform-{F}ourier analysis for nonlinear
  problems.
\newblock {\em Stud. Appl. Math.}, 53(4):249--315 (1974)

\bibitem{asy}
Adler V.E., Shabat A.B., Yamilov R.I.:
\newblock Symmetry approach to the integrability problem.
\newblock {\em Theor. Math. Phys.}, 125(3):1603--1661 (2000)

\bibitem{mr99i:35005}
Beukers F., Sanders J.A., Wang J.P.:
\newblock One symmetry does not imply integrability.
\newblock {\em J. Differential Equations}, 146(1):251--260 (1998)

\bibitem{mr88g:58195}
 Chen H.~H.,  Lin J.~E.:
\newblock On the infinite hierarchies of symmetries and constants of motion for
  the {K}adomtsev-{P}etviashvili equation.
\newblock {\em Phys. D}, 26(1-3):171--180 (1987)

\bibitem{mr85m:58187}
 Chen H.~H., Lee Y.~C., Lin J.E.:
\newblock On a new hierarchy of symmetries for the {K}adomtsev-{P}etviashvili
  equation.
\newblock {\em Phys. D}, 9(3):439--445 (1983)

\bibitem{mr93c:58085}
 Dorfman I.~Y.,  Fokas A.~S.:
\newblock Hamiltonian theory over noncommutative rings and integrability in
  multidimensions.
\newblock {\em J. Math. Phys.}, 33(7):2504--2514 (1992)

\bibitem{djm82}
 Date E., Jinbo M.,  Miwa T.:
\newblock Method for generating discrete soliton equations. {I}.
\newblock {\em Journal of the Physical Society of Japan},
51(12):4116--4124 (1982)

\bibitem{mr94j:58081}
 Dorfman I.:
\newblock {\em Dirac structures and integrability of nonlinear evolution
  equations}.
\newblock John Wiley \& Sons Ltd., Chichester (1993)

\bibitem{DS08}
 Demskoi D.K., Sokolov  V.V.:
\newblock On recursion operators for elliptic models.
\newblock {\em Nonlinearity}, 21:1253--1264 (2008)

\bibitem{mr83f:58039}
 Fokas A.~S., Fuchssteiner B.:
\newblock The hierarchy of the {B}enjamin-{O}no equation.
\newblock {\em Phys. Lett. A}, 86(6-7):341--345 (1981)

\bibitem{ffokas02}
Finkel F., Fokas A.S.:
\newblock On the construction of evolution equations admitting a master
  symmetry.
\newblock {\em Physics Letters A}, 293(1-2):36--44 (2002)

\bibitem{zhang13}
Fu W., Huang L.,  Tamizhmani K.~M., Zhang D.J.:
\newblock Integrability properties of the differential-difference
  {K}adomtsev-{P}etviashvili hierarchy and continuum limits.
\newblock {\em Nonlinearity}, 26(12):3197 (2013)

\bibitem{mr1465768}
Fuchssteiner B., Ivanov S., Wiwianka W.:
\newblock Algorithmic determination of infinite-dimensional symmetry groups for
  integrable systems in $1+1$ dimensions.
\newblock {\em Math. Comput. Modelling}, 25(8-9):91--100 (1997)

\bibitem{FNR13}
Ferapontov E.V., Novikov V.S., Roustemoglou I.:
\newblock Towards the classification of integrable differential-difference
  equations in $2+1$ dimensions.
\newblock {\em Journal of Physics A: Mathematical and Theoretical},
  46(24):245207 (2013)

\bibitem{Fuc79}
Fuchssteiner, B.:
\newblock Application of hereditary symmetries to nonlinear evolution
  equations.
\newblock {\em Nonlinear Analysis, Theory, Methods {\rm{\&}} Applications},
  3(11):849--862 (1979)

\bibitem{mr86c:58158}
Fuchssteiner, B.:
\newblock Mastersymmetries, higher order time-dependent symmetries and
  conserved densities of nonlinear evolution equations.
\newblock {\em Progr. Theoret. Phys.}, 70(6):1508--1522 (1983)

\bibitem{Fuch84}
Fuchssteiner, B.:
\newblock On the hierarchy of the {L}andau-{L}ifshitz equation.
\newblock {\em Physica D: Nonlinear Phenomena}, 13(3):387--394 (1984)

\bibitem{HNov13}
Huard B.,  Novikov V.S.:
\newblock On classification of integrable {D}avey-{S}tewartson type equations.
\newblock {\em Journal of Physics A: Mathematical and Theoretical},
  46(27):275202 (2013)

\bibitem{humph2008}
Humphreys, J.E.:
\newblock {\em Representations of Semisimple Lie Algebras in the {BGG} Category
  $\cal{O}$}, volume~94 of {\em Graduate studies in mathematics}.
\newblock American Mathematical Society (2008)

\bibitem{farbod14}
Khanizadeh, F.:
  \newblock The master symmetry and time dependent symmetries of the
  differential–difference {KP} equation.
\newblock {\em Journal of Physics A: Mathematical and Theoretical},
  47(40):405205  (2014)

\bibitem{LM04}
Lombardo S.,  Mikhailov A.V.:
\newblock Reductions of integrable equations: dihedral group.
\newblock {\em Journal of Physics A: Mathematical and General}, 37:7727--7742
 (2004)

\bibitem{ma92}
 Ma, W.X.:
\newblock Lax representations and {L}ax operator algebras of isospectral and
  nonisospectral hierarchies of evolution equations.
\newblock {\em Journal of Mathematical Physics}, 33(7):2464--2476 (1992)

\bibitem{Mag80}
Magri, F.:
\newblock A geometrical approach to the nonlinear solvable equations.
\newblock volume 120 of {\em Lecture Notes in Physics}, pages 233--263.
  Springer--Verlag (1980)

\bibitem{mazor2010}
Mazorchuk, V.:
\newblock {\em Lectures on $\mathfrak{sl}_2(\mathbb{C})$-modules}.
\newblock Imperial College Press (2010)

\bibitem{mik79}
 Mikhailov, A.V.:
\newblock Integrability of a two-dimensional generalization of the {T}oda
  chain.
\newblock {\em JETP Lett.}, 30(7):414--418 (1979)

\bibitem{mik80}
 Mikhailov, A.V.:
\newblock Reduction in integrable systems. {T}he reduction group.
\newblock {\em JETP Lett.}, 32(2):187--192 (1980)

\bibitem{mik81}
 Mikhailov, A.V.:
\newblock The reduction problem and the inverse scattering method.
\newblock {\em Phys. D}, 3(1\& 2):73--117, 1981.

\bibitem{MaN01}
 Maltsev A.Ya,  Novikov S.P.:
\newblock On the local systems {H}amiltonian in the weakly nonlocal {P}oisson
  brackets.
\newblock {\em Physica D: Nonlinear Phenomena}, 156(1-2):53--80 (2001)

\bibitem{mn2}
 Mikhailov A.V., Novikov V.S.:
\newblock Classification of integrable {B}enjamin-{O}no-type equations.
\newblock {\em Moscow Mathematical Journal}, 3(4):1293--1305  (2003)

\bibitem{mnw08}
 Mikhailov A.V., Novikov V.S.,  Wang J.P.:
\newblock Symbolic representation and classification of integrable systems.
\newblock In M.A.H. MacCallum and A.V. Mikhailov, editors, {\em Algebraic
  Theory of Differential Equations}, pages 156--216. Cambridge University
  Press (2009)

\bibitem{MPW}
 Mikhailov A.V., Papamikos G.,  Wang J.P.:
\newblock Darboux transformation with dihedral reduction group.\newblock {\em
Journal of Mathematical Physics}, 55(11):113507 (2014)

\bibitem{mr93b:58070}
 Mikhailov A.V., Shabat A.~B., Sokolov V.~V.:
\newblock The symmetry approach to classification of integrable equations.
\newblock In {\em What is integrability?}, Springer Ser. Nonlinear Dynamics,
  pages 115--184. Springer, Berlin (1991)

\bibitem{mr89e:58062}
 Mikhailov A.V., Shabat A.~B., Yamilov R.~I.:
\newblock A symmetry approach to the classification of nonlinear equations.
  {C}omplete lists of integrable systems.
\newblock {\em Uspekhi Mat. Nauk}, 42(4(256)):3--53 (1987)

\bibitem{mr89g:58092}
 Mikhailov A.V., Shabat A.~B., Yamilov R.~I.:
\newblock Extension of the module of invertible transformations.
  {C}lassification of integrable systems.
\newblock {\em Comm. Math. Phys.}, 115(1):1--19 (1988)

\bibitem{mr1643816}
 Mikhailov A.V.,  Yamilov R.~I.:
\newblock Towards classification of $(2+1)$-dimensional integrable equations.
  {I}ntegrability conditions. {I}.
\newblock {\em J. Phys. A}, 31(31):6707--6715 (1998)

\bibitem{oef82}
Oevel W., Fuchssteiner B.:
\newblock Explicit formulas for symmetries and conservation laws of the
  {K}adomtsev-{P}etviashvili equation.
\newblock {\em Physics Letters A}, 88:323--327 (1982)

\bibitem{mr90j:58061}
Oevel W., Fuchssteiner B., Zhang H., Ragnisco O.:
\newblock Mastersymmetries, angle variables, and recursion operator of the
  relativistic {T}oda lattice.
\newblock {\em J. Math. Phys.}, 30(11):2664--2670 (1989)

\bibitem{mr58:25341}
 Olver, P.~J.:
\newblock Evolution equations possessing infinitely many symmetries.
\newblock {\em J. Mathematical Phys.}, 18(6):1212--1215 (1977)

\bibitem{olvgho}
 Olver, P.~J.:
\newblock Nonlocal symmetries and ghosts.
\newblock In  Shabat A.B., ~Gonz{\'a}lez-L{\'o}pez A., ~Ma{\~n}as M.,
  ~Mart{\'i}nez~Alonso L.,  Rodr{\'i}guez M.A.(eds), {\em New Trends in
  Integrability and Partial Solvability}, volume 132 of {\em NATO Science
  Series}, pp. 199--215. Springer Netherlands (2004)

\bibitem{OS86}
 Orlov A.Yu., Schulman  E.I.:
\newblock Additional symmetries for integrable equations and conformal algebra
  representation.
\newblock {\em Letters in Mathematical Physics}, 12(3):171--179 (1986)

\bibitem{mr99c:58077}
 Olver P.~J., Sokolov V.V.:
\newblock Integrable evolution equations on associative algebras.
\newblock {\em Comm. Math. Phys.}, 193(2):245--268 (1998)

\bibitem{mr1900193}
 Olver P.~J., Sanders J.A., Wang J.P.:
\newblock Ghost symmetries.
\newblock {\em J. Nonlinear Math. Phys.}, 9(suppl. 1):164--172, 2002.
\newblock Recent advances in integrable systems (Kowloon, 2000).

\bibitem{mr1781148}
 Olver P.~J., Wang J.P.:
\newblock Classification of integrable one-component systems on associative
  algebras.
\newblock {\em Proc. London Math. Soc. (3)}, 81(3):566--586 (2000)

\bibitem{mr95j:35010}
Shabat A.B.,  Mikhailov A.~V.:
\newblock Symmetries--test of integrability.
\newblock In {\em Important developments in soliton theory}, pp. 355--374.
  Springer, Berlin (1993)

\bibitem{mr86i:58070}
Sokolov V.V., Shabat A.B.:
\newblock Classification of integrable evolution equations.
\newblock In {\em Mathematical physics reviews, Vol. 4}, volume~4 of {\em
  Soviet Sci. Rev. Sect. C: Math. Phys. Rev.}, pp. 221--280. Harwood Academic
  Publ., Chur, (1984)

\bibitem{mr1952867}
Sergyeyev A., Sanders J.A.:
\newblock A remark on nonlocal symmetries for the
  {C}alogero-{D}egasperis-{I}bragimov-{S}habat equation.
\newblock {\em J. Nonlinear Math. Phys.}, 10(1):78--85 (2003)

\bibitem{MR99g:35058}
Sanders J.A., Wang J.P.:
\newblock On the integrability of homogeneous scalar evolution equations.
\newblock {\em J. Differential Equations}, 147(2):410--434 (1998)

\bibitem{mr1874291}
Sanders J.A., Wang J.P.:
\newblock On integrability of evolution equations and representation theory.
\newblock In {\em The geometrical study of differential equations (Washington,
  DC, 2000)}, pp. 85--99. Amer. Math. Soc., Providence, RI (2001)

\bibitem{mr2002b:37100}
Sanders J.A., Wang J.P.:
\newblock On recursion operators.
\newblock {\em Phys. D}, 149(1-2):1--10 (2001)

\bibitem{sw2009}
Sanders J.A., Wang J.P.:
\newblock Number theory and the symmetry classification of integrable systems.
\newblock In Mikhailov, A.V.(ed) {\em Integrability}, volume 767 of
  {\em Lecture Notes in Physics}, pp. 89--118. Springer Berlin Heidelberg,
  (2009)

\bibitem{troost2012}
Troost,J.:
\newblock Models for modules: {\em The story of} {$\mathcal{O}$}.
\newblock {\em Journal of Physics A: Mathematical and Theoretical},
  45(41):415202 (2012)

\bibitem{wang21}
Wang, J.P.:
\newblock On the structure of $(2+1)$-dimensional commutative and
  noncommutative integrable equations.
\newblock {\em J. Math. Phys.} 47(11):113508 (2006)

\bibitem{wang09}
Wang, J.P.:
\newblock Lenard scheme for two-dimensional periodic Volterra chain.
\newblock {\em J. Math. Phys.} 50:023506 (2009)

\end{thebibliography}

\end{document}